%% file: main.tex
\documentclass{article}
\usepackage{authblk}

\usepackage{amsmath, amssymb, amsthm}
\usepackage{cite}
\usepackage{ulem}

\usepackage{amsthm, amsfonts, mathrsfs}
\usepackage{ dsfont }
\usepackage{amssymb}
\usepackage{enumerate}
\usepackage{graphicx}
\usepackage{subfig}

\usepackage{float}
\usepackage{bbm}
\usepackage{amsmath}
\usepackage{comment}
\usepackage{hyperref}
\usepackage{listings}
\usepackage{color}

\usepackage{float}

\usepackage{algorithm}
\usepackage{algorithmic}

\usepackage{longtable}
\usepackage{booktabs}

\usepackage[colorinlistoftodos,prependcaption,textsize=tiny]{todonotes}

\newcommand{\R}{\mathbb{R}}

\newcommand{\N}{\mathbb{N}}

\newtheorem{theorem}{Theorem}[subsection]
\newtheorem{lemma}[theorem]{Lemma}  
\newtheorem{proposition}[theorem]{Proposition}
\newtheorem{claim}{Claim}
\newtheorem{remark}[theorem]{Remark}
\newtheorem{problem}[theorem]{Problem}
\newtheorem{definition}[theorem]{Definition}

\title{An Efficient Transport-Based Dissimilarity Measure for Time Series Classification under Warping Distortions}

\author[1]{Akram Aldroubi}
\author[2]{Rocío Díaz Martín}
\author[3]{Ivan Medri}
\author[3]{Kristofor E. Pas}
\author[3]{Gustavo K. Rohde}
\author[4]{Abu Hasnat Mohammad Rubaiyat}

\affil[1]{Department of Mathematics, Vanderbilt University, Nashville, TN, USA\\
Email: aldroubi@vanderbilt.edu}

\affil[2]{Department of Mathematics, Tufts University, Medford, MA, USA\\
Email: rocio.diaz\_martin@tufts.edu}

\affil[3]{Department of Biomedical Engineering, University of Virginia, Charlottesville, VA, USA\\
Emails: \{pzr7pr, xfd3sy, gustavo\}@virginia.edu}

\affil[4]{Jacobs, Hanover, MD, USA\\
Email: ar3fx@virginia.edu}

\begin{document}

\maketitle

\begin{abstract} 
    Time Series Classification (TSC) is an important problem with numerous applications in science and technology. Dissimilarity-based approaches, such as Dynamic Time Warping (DTW), are classical methods for distinguishing time series when time deformations are confounding information. In this paper, starting from a deformation-based model for signal classes we define a problem statement for time series classification problem. We show that, under theoretically ideal conditions, a continuous version of classic 1NN-DTW method can solve the stated problem, even when only one training sample is available. In addition, we propose an alternative dissimilarity measure based on Optimal Transport and show that it can also solve the aforementioned problem statement at a significantly reduced computational cost. Finally, we demonstrate the application of the newly proposed approach in simulated and real time series classification data, showing the efficacy of the method.
\end{abstract}

\section{Introduction}
Developing accurate and interpretable time series predictive models is a foundational goal in data science. These models have a variety of uses, such as anomaly detection \cite{rebbapragada2009finding}, human activity recognition \cite{lara2012survey}, communications \cite{fehske2005new}, structural health monitoring \cite{abdeljaber20181}, and  others \cite{BISVP08, berkaya2018survey, rubaiyat2024end, subasi2010eeg}. One popular problem is Time Series Classification (TSC), where given a signal, $s(t)$, the goal is to distinguish its class. Numerous TSC methods have been developed (see \cite{bagnall2017great,middlehurst2024bake} for recent reviews). They are drawn upon three foundational techniques: (1) feature-based, (2) deep learning-based,  and (3) distance-based methods. 

Feature-based techniques extract information from signals as features. Example features range in complexity, from fundamental statistics to polynomial approximations. However, discerning meaningful features without prior knowledge is a difficult problem. To overcome this, algorithms have been developed to capture an array of feature candidates. Two examples of this include the \textit{hctsa} toolbox \cite{fulcher2017hctsa} and TSFresh\cite{christ2018time}. While different in application (i.e., exploratory analysis vs. preprocessing for regression models), these provide a valuable baseline for exploring potentially discriminative information between classes. Building upon this, methods have been proposed to further improve these techniques through computing notable canonical variables \cite{lubba2019catch22canonicaltimeseriescharacteristics} and using {random} forest algorithms to improve adaptability \cite{middlehurst2022freshprince}. Feature-based methods for classifying signals might provide insights (\textit{a posteriori}) into how classes may deviate from one another. However, deriving mathematical models in feature methods have limitations in their physical interpretability, since any linear or nonlinear combination of features with different units do not have a straighforward interpretation. Therefore, any feature space distances derived from an ad hoc features is also difficult to interpret. Moreover,  within the context of feature selection, potential bias through user-input is a major factor in model quality.


Deep learning methods have been proposed as an end-to-end approach for TSC with increasing popularity. These models work through minimizing user input in favor of parameter tuning to maximize accuracies on a training dataset. Because of this,  deep learning often lacks interpretability. Nevertheless, if the model accurately captures the decision boundary between classes, its application post-training is exceptionally straight forward. At its core, many deep learning methods are based on Convolutional Neural Networks (CNNs) with primary modifications being derived on model architecture or filters within the model design. Notable examples include ResNet methods for classifying time series \cite{wang2016timeseriesclassificationscratch,cheng2021time,wu2019resnet} which uses sequential convolutional layers to extract features prior to the network, with varying components involved in further processing. Other more novel approaches include InceptionTime \cite{ismail2020inceptiontime}, which uses inception modules for feature extraction, and has a further derivative of H-InceptionTime \cite{10020496} using one dimensional filters for further feature extraction.

Finally, distance-based techniques work with raw time series and define similarity measures 
that are often used with conventional machine learning. These methods are attractive for modeling, due to their simple mathematical structure. For time series, elastic distances are most widely used for their ability to overcome confounding deformations, with dynamic time warping (DTW) being a prominent example \cite{sakoe1978dynamic}. Elastic measures implicitly assume that signals from each class are instances of certain (usually unknown) templates that suffer time warps or deformations. This formulation is intuitive for physical systems, i.e., deformations can include delays, frequency modulation or changes in speed or acceleration, for which we may require  translation and/or dilation invariant models or metrics \cite{benamou2000computational}. Computational methods for calculating DTW and other elastic distances have been extensively studied \cite{buchin2007computability, buchin2009exact,mueen2018speeding,li2018similarity}. Among current approaches to time series classification, simple methods such as nearest neighbor classifiers using DTW-based divergences are still current and among the most accurate \cite{bagnall2017great,middlehurst2024bake}. As a result, DTW remains an appealing standard for TSC with one of its main disadvantages being the lack o closed-form expression, relying on approximations via dynamic programming algorithms, which come at a high computational cost.\\

In this article, we examine in more depth the mathematical structure of distance-based classifiers under deformation-based model for each class. Although the deformation model is implicit in the underlying generative model, the precise formulation of this is not generally established in the literature. We show that DTW-related techniques give rise to an equivalence relation in a certain space of functions (Proposition \ref{prop: equivalence relation}), allowing the user to utilize any class sample as a template. For this,  we use a continuous version of DTW, referred to by us as CDTW, to avoid discretization artifacts (Section \ref{sec: CDTW}, Theorem \ref{thm: cdtw class}). We explain the connection between continuous and discrete formulations in Section \ref{sec: dtw - cdtw} (Theorem \ref{thm: CDTW-DTW param space}). The resulting equivalence classes will partition the space of time series and determine a precise mathematical class model (Section \ref{sec: model class}).
Finally, we propose an alternative and faster divergence $d_T$ following a similar class formulation (Section \ref{subsec: sdb}). 
We call $d_T$ a transport divergence since it can be computed using closed formulas related to the optimal transport (OT) theory. We show the relation between $d_T$ and CDTW (Theorem \ref{thm: cdtw implies sdb}). We show theoretical guarantees for classification using CDTW and $d_T$ (Theorem \ref{thm: CDTW guarantees}) and expand this with experiments in simulated and real world datasets (Section \ref{sec: experiments}). \\

\textbf{{Summary of contributions:}} 
\begin{itemize}
    \item We define a model for classes in 1D signal classification problems inspired in time warping techniques. We show that the classes arise from an equivalence relation. Section \ref{sec: model class}.
    \item We show that the equivalence classes are given as the 0-level sets of a continuous version of DTW. That is,  $s\sim \phi \iff CDTW(s,\phi) = 0$. Thus, making the classification problem well-posed, and exactly solvable. Theorems \ref{thm: cdtw class} and \ref{thm: CDTW guarantees}.
    \item We propose a new faster divergence that approximately follows the class model and for which accuracy guarantees can also be obtained. Definition \ref{def sdb} and Theorem \ref{thm: CDTW guarantees}.
\end{itemize}

\medskip

\textbf{{Organization of the paper:}} 
In Section \ref{sec: prelim} we review the necessary definitions of DTW and OT, which form the foundation of our proposed class models and classification algorithms.
Sections \ref{sec: model} to \ref{sec: experiments} have our new results and experiments as outlined 
above. Section \ref{sec: proofs} has the proofs of all our new results.

\medskip

\textbf{Keywords:} Signal Classification, Continuous Dynamic Time Warping, Optimal Transport, Numerosity Reduction.

\section{Preliminaries}\label{sec: prelim}
In this section, we explicitly state the DTW divergence. Also,  we review some basic facts from Transport Theory for 1D probability distributions (see \cite{Villani2003Topics, Villani2009Optimal, Santambrogio-OTAM, Thorpe2018Introduction}). 
   \subsection{The DTW divergence}\label{sec: dtw}

    \begin{definition}[DTW]\label{def: DTW} The  Dynamic Time Warping (DTW) divergence between
    two discrete functions $\overline{s}:[0,1,\dots,N] \to \R$, $\overline{\phi}:[0,1,\dots,M] \to \R$ is defined as
    \begin{equation}\label{eq: DTW problem}
        DTW(\overline{s},\overline{\phi}) = \min_{\overline{g}_1,\overline{g}_2}\sum_{i=0}^L  |\overline{s}(\overline{g}_1(i)) -\overline{\phi}(\overline{g}_2(i) )|         
    \end{equation}
    where the minimization is over all {discrete time warps (deformations)}  $\overline{g}_1, \overline{g}_2: [0,1,\dots,L] \to \mathbb{N}$, for some $L\in \N$ with $\min\{N,M\}\leq L \leq M+N $, such that
    \begin{itemize}
        \item $\overline{g}_1(0) = 0 = \overline{g}_2(0)$
        \item $\overline{g}_1(L) = N$ and         $\overline{g}_2(L) = M$
        \item $0 \leq \overline{g}_i(j+1) - \overline{g}_i(j) \leq 1$
    \end{itemize}    
    For any admissible $\overline{g}_1,\overline{g}_2$ we can define a finite sequence $\overline{\gamma}: [0,1,\dots,L]\to \N\times \N$ as $\overline \gamma = (\overline g_1,\overline g_2)$. We say that $\overline{\gamma}$ is a discrete warping path. The warping path that attains the infimum in \eqref{eq: DTW problem} is called an optimal discrete warping path between $s$ and $\phi$. 
    \end{definition}
    Note from the last condition that the time warps $\overline{g}_i$ must be non-decreasing.
    For simplicity, in this article we will only consider the case $N=M$.\\
    
   \textbf{Computational cost:} The time complexity of the DTW algorithm is of order $\mathcal{O}(N^2)$ for two input sequences of length $N$, and it can be optimized to $\mathcal{O}({N^{2}}/\log( \log (N)))$ \cite{gold2016dynamic}.

\subsection{One-dimensional Optimal Transport}

    Given probability measures $\mu$ and $\nu$ on the real line, the 2-Wasserstein distance is defined by  
    \begin{equation}\label{eq: wass1}
        W_2(\mu,\nu):=\inf_{\pi} \left(\int |x-y|^2 d\pi(x,y)\right)^\frac{1}{2}
    \end{equation}
    where the infimum is taken over all \textit{transport plans} $\pi$ which are probability measures on $\mathbb{R}\times \mathbb{R}$ with first and second marginals $\mu$ and $\nu$, respectively. This definition extends to $\mathbb{R}^n$, and such an infimum (under certain decay conditions) is a finite number and is always attained.  
    
    In the case of considering $\mu,\nu$ one-dimensional probability distributions, 
    let $F_{\mu}$ and $F_{\nu}$ denote their cumulative distribution functions (CDFs), and consider $F_{\mu}^{\dagger}$ and $F_{\nu}^{\dagger}$ as their generalized inverses (or ``quantile functions'') defined below:  
    
    \begin{definition}[Generalized inverse - Comments and properties in Section \ref{app: gen inv}]\label{def: gen inv}  
        Let $F:[a,b] \to [c,d]$ be a non-decreasing, right-continuous function.  
        We define the generalized inverse of $F$, denoted $F^\dagger$, as the function $F^\dagger:[c,d] \to [a,b]$ given by  
        \begin{equation*}
            F^\dagger(y) := \inf \{x\in [a,b]: F(x) > y \}, 
        \end{equation*}
        where we impose $\inf \emptyset = b$.  
    \end{definition}  
    
    The 2-Wasserstein distance \eqref{eq: wass1} can be expressed as  $W_2(\mu,\nu)^2=\int_0^1 |F_{\mu}^{\dagger}(x)-F_{\nu}^{\dagger}(x)|^2dx.$
    Moreover, if at least one of the two distributions does not give mass to atoms, say $\mu$, then the optimal transport plan is unique and supported on the graph of the function 
    \begin{equation}\label{eq: opt map}
        g=F_\nu^{\dagger}\circ F_\mu
    \end{equation} 
    called the \textit{optimal transport map} from $\mu$ to $\nu$. The condition on the marginals of the coupling $\pi$ then translates into saying that the function $g$ pushes the measure $\mu$ onto the measure $\nu$ by the formula 
    \begin{equation}\label{eq: push}
    \nu(A) = \mu\left(\{x\mid \, g(x)\in A\}\right) \quad \text{for every measurable set } A.    
    \end{equation}
    When the measures have densities, that is, when $d\mu = a(x)dx$, $d\nu = b(x)dx$, the push-forward relation \eqref{eq: push} can be recast as 
    \begin{equation}\label{eq: change var}
        a(x) = b(g(x)) g'(x).
    \end{equation}
    This indicates that the transport map $g$ also acts as a ``time warp'' that aligns the density functions $a$ and $b$. However, the matching assumes a normalization condition, that is, both functions $a,b$ should be non-negative and integrate one.
    In equation \eqref{eq: change var}, this condition is reflected by the multiplication with $g'(x)$, which arises from the change-of-variables formula and ensures that total probability mass is preserved. As such, $g$ provides a registration between $a$ and $b$, with the extra factor $g'$ accounting for the change in density induced by the transformation.\\
    
   \textbf{Computational cost:} Computing the transport map $g$ for one-dimensional signals involves evaluating two integrals to obtain the cumulative distribution functions (CDFs), followed by an interpolation step to compute the inverse of one CDF. Assuming the signals are discretized with $N$ points, and using cumulative sums for the CDFs and linear interpolation for the inverse, the total time complexity is $\mathcal{O}(N \log (N))$.

\section{Model Class and Classification Problem}\label{sec: model}
    In this section, we state a classification problem with a particular class model. We show that the problem is well-posed. Inspired by the DTW method and the theory of Optimal Transport, we introduce the CDTW and $d_T$ divergences, respectively, and relate classes with level sets of these divergences. This will allow us to find closed formulas to classify signals under this model in later sections.
    
\subsection{Proposed Model Class}\label{sec: model class}
    The 1D \textit{generative model} for class $c$ generated by \textit{template} functions $\{\phi_m^{(c)}\}_{m=1}^{M_c}$ is defined to be the set
        \begin{equation} \label{eq: full class model}
            \mathbb{S}^{(c)} = \bigcup_{m=1}^{M_c} \mathbb{S}_{\phi_m^{(c)}},
        \end{equation}
        where 
        \begin{equation}\label{eq: atomic class}
            \mathbb{S}_{\phi_m^{(c)}} = \left\{ s\mid \quad s\circ g_1 = \phi_m^{(c)}\circ g_2 \quad \text{ for some } g_1,g_2 \in \mathcal{G}\right\}
        \end{equation}
        for some predefined set of functions $\mathcal{G}$, which is common for all the classes. 
        We call the subclass $\mathbb{S}_{\phi_m^{(c)}}$ as the \textit{atomic class} associated with template $\phi_m^{(c)}$.\\

        \textbf{Standing assumptions.} In this work, certain assumptions are needed for all  functions $s,\phi,g_1,g_2$, etc. {These are detailed in the definition of what we refer to as hypothesis \textbf{H} in Section \ref{app: model class}, Definition \ref{def: hyp H}.}
        For simplicity, we will consider signals $s,\phi:[0,1]\to\R$, but the proposed analysis can easily be extended to signals defined over any interval  $[a,b]$. Some regularity condition are also needed. To improve readability,  they are defered to  Definition \ref{def: C1_FF} in Section \ref{app: model class}.
        The set $\mathcal{G}$ needs additional properties that are explained in Definition \ref{def: G} of Section \ref{app: model class}. It is motivated by the goal of providing a continuous counterpart to the functions $\overline{g}_1$ and $\overline{g}_2$ in the classical definition of the DTW method. Indeed, $\mathcal{G}$ is the set of non-decreasing functions $g:[0,1]\to \R$ that satisfy $g(0)=0$, $g(1)=1$, with additional regularity conditions. We add the qualifier of \textbf{informal statement} to any definition or result where we do not explicitly state the regularity conditions of hypothesis $\textbf{H}$. Nevertheless, full requirements are given in the proofs. 
        
    \begin{problem}[Classification problem] \label{def: classification problem}
        Let classes be defined by $\mathbb{S}^{(c)}$, $c=1,\dots,N$. Given a set of labeled training samples $\{s_1,s_2,\dots\}$, determine the class of an unclassified signal $s$.
    \end{problem}
    \begin{proposition}[Informal statement - Proof in Section \ref{app: Equivalence relation}]\label{prop: equivalence relation}
        Under an specific choice of the set of function $\mathcal{G}$, the atomic class $\mathbb{S}_\phi$ is exactly the equivalence class of the function $\phi$ given by the relation $s\sim \phi$ if and only if there exists $g_1$, $g_2\in \mathcal{G}$ such that $s\circ g_1 = \phi \circ g_2$.
    \end{proposition}
    Two important consequences can be derived from this proposition. One is that having equivalence classes imply that different classes are automatically disjoint, making the classification problem well-posed. The other one is that since $\mathbb{S}_\phi$ is an equivalence class, a closed form solution to the mathematical classification problem statement above can be found with only one training sample per atomic class (see Theorem \ref{thm: CDTW guarantees}).

    \subsection{The CDTW divergence}\label{sec: CDTW}
    \begin{definition}[Informal definition CDTW]\label{def: CDTW}
    Given functions $s:[0,1]\to \mathbb R$, $\phi:[0,1]\to \mathbb R$, Continuous Dynamic Time Warping (CDTW) dissimilarity between $s$ and $\phi$ is defined as
     \begin{equation}\label{eq: cdtw12}
        CDTW(s,\phi):=\inf_{\gamma}\int_0^1
        h(\gamma(t))\|\gamma'(t)\|_1dt
    \end{equation}
    where the optimization is over all the \textit{continuous warping paths} (curves) $\gamma:[0,1]\to[0,1]\times [0,1]$ of the form $\gamma(t)=(g_1(t),g_2(t))$ such that its coordinate functions $g_1,g_2:[0,1]\to \R$ belong to $\mathcal{G}$; that is, they satisfy the following conditions together with some extra regularity assumptions (see Definition \ref{def: G}),
    \begin{itemize}
        \item $g_1(0)=0=g_2(0)$
        \item $g_1(1)=1=g_2(1)$
        \item are non-decreasing functions 
    \end{itemize}
    and where $\|\gamma'(t)\|_1:=|g_1'(t)|+|g_2'(t)|$ and the cost function $h$ is  $h(\gamma(t)):=|s(g_1(t))-\phi(g_2(t))|$.
    \end{definition}
    In the above definition the curves $\gamma(t) = (g_1(t),g_2(t))$ represent the proposed alignments between points in the domain of $s$ and $\phi$ , analogously to the functions $\overline{g_1},\overline{g_2}$ in Definition \ref{def: DTW}. The function $h$ is the mismatch between $s,\phi$ under said alignment. The total integral is the line integral of the mismatch along the curve $\gamma$, which account for the total mismatch standardize by arc-length.  
    \begin{theorem}[Informal statement - Proof in Section \ref{app: CDTW}]\label{thm: cdtw class}\, 
        Assume that $CDTW(s,\phi)=0$ and that the infimum in the definition of $CDTW$ is attained at a path $\gamma^*(t):= (g_1(t),g_2(t))$. Then, $s\circ g_1=\phi\circ g_2$. 
        Reciprocally, if there exists $g_1,g_2\in\mathcal{G}$ such that $s\circ g_1=\phi\circ g_2$, then $CDTW(s,\phi)=0$ and the minimum is attained by the curve $\gamma^*(t):=(g_1(t),g_2(t))$.
    \end{theorem}

    The importance of Theorem \ref{thm: cdtw class} is that $CDTW(s,\phi)=0$ characterizes atomic class $\mathbb S_{\phi}$ defined by expression \eqref{eq: atomic class}.

    \subsubsection{CDTW computation: approximation with DTW}\label{sec: dtw - cdtw}

    CDTW is, in general, very difficult to compute exactly. The main advantage of the continuous formulation is theoretical since it provides an exact representation of the model classes that does not depend on the choice of discretization. In our experiments, we approximate CDTW using a weighted DTW. The justification for this choice is given in Theorem \ref{thm: CDTW-DTW param space} below, which claims that under a sufficiently fine discretization of the signals, DTW is a good enough approximation of the proposed CDTW. 
    
    Let us introduce the \textit{weighted} version of Definition \ref{def: DTW}:
    \begin{equation}\label{eq: DTW problem w}
        DTW_{\omega}(\overline{s},\overline{\phi}) := \min_{\overline{g}_1,\overline{g}_2}\sum_{i=0}^L  |\overline{s}(\overline{g}_1(i)) -\overline{\phi}(\overline{g}_2(i) )|\omega_i         
    \end{equation}
    where 
    the discrete warping path $\overline \gamma = (\overline{g}_1,\overline{g}_2)$ is as in Definition \ref{def: DTW},
    and the vector of weights $\omega$ depends on $\overline \gamma$
    and is defined by
    \begin{equation}\label{eq: weight}
            \omega_i:=\frac{|\overline{g}_1(i)-\overline{g}_1(i-1)|+|\overline{g}_2(i)-\overline{g}_2(i-1)|}{N}.
    \end{equation}
    Notice that if we
    consider the $[0,N]\times[0,N]$ square, and subdivide it into $N\times N$ unit sub-squares, then the points $\{(\overline{g}_1(i),\overline{g}_2(i))\}_{i=0}^L$ coincide with some vertices of these sub-squares, starting from $(0,0)$ and ending at $(N,N)$. By connecting these vertices through horizontal, vertical, or diagonal segments, we observe that 
    $\omega_i$ is either 1 (if the corresponding segment is horizontal or vertical) or 2 (if diagonal).

    For the next theorem we need some terminology. We say that \textit{the uniform partition} of an interval $[0,1]$ with step size $\Delta$ is the set of points $u_j=j\Delta, \text{ where } \Delta = \frac{1}{N} \text{ for some } N\in\mathbb{N} \text{ and } j\in\{0,\dots,N\}$. Also, we say that \textit{a discretization of $s:[0,1]\to \mathbb{R}$ over the partition} $\{u_j\}_{j=0}^N$ of $[0,1]$ is the vector $\overline{s}\in \mathbb{R}^{N+1}$ given by $\overline{s}(j) = s(u_j)$. 

    \begin{theorem}[CDTW -- DTW - Proof in Section \ref{app: CDTW-DTW}]\label{thm: CDTW-DTW param space} Let $s,\phi:[0,1]\to \mathbb R$. For every $\varepsilon>0$, there exists $\delta>0$ such that if we consider a uniform partition of $[0,1]$ with step-size $\Delta \leq \delta$, then
        \begin{equation*}
            |CDTW(s,\phi)-DTW(\overline{s}, \overline{\phi})|<\varepsilon
        \end{equation*}
    for $\overline{s},\overline{\phi}\in \mathbb R^{N+1}$  the corresponding discretizations of $s$ and $\phi$.
    \end{theorem} 

    \subsection{The $d_T$ Divergence}\label{subsec: sdb} 
    \begin{definition}\label{def sdb}
        Given two differentiable signals $s,\phi:[0,1]\to\mathbb{R}$ we define the following transport-like dissimilarity 
        \begin{equation}
            d_{T}(s,\phi) = \left(\int_0^1 |s(x)-\phi(g^s_\phi(x))|^2 \sqrt{(g^s_\phi)'(x)}\ dx \right)^{1/2}
        \end{equation}
        where $g^s_\phi$ is the optimal transport map between the probability density functions $|s'|/\|s'\|_1$ and $|\phi'|/\|\phi'\|_1$, that is, it  is characterized by the closed-formula    
        \begin{equation}\label{eq: g_explicit}
            g^s_\phi=F_{\frac{|\phi'|}{\|\phi'\|_1}}^{\dagger}\circ F_{\frac{|s'|}{\|s'\|_1}}.
        \end{equation}
    \end{definition}
    
    The motivation of Definition \ref{def sdb} comes from doing some formal manipulations on the equivalence relation $s\sim \phi$ if and only if $s\circ g_1 = \phi\circ g_2  $. Specifically, differentiating on both sides of the equality, and using that $g_1, g_2$ are non-decreasing we get $(|s'|\circ g_1) g_1' = (|\phi'|\circ g_2) g_2'$. Thus, denoting $F_{|s'|}(t) = \int_{-\infty}^t |s'(x)| dx$ (resp for $F_{|\phi'|}$), and integrating we get 
    \begin{align*}
        F_{|s'|}(g_1(t)) = F_{|\phi'|}(g_2(t))\qquad
        \text{and} \qquad
        g_2\circ g_1^{-1} = F_{|\phi'|}^{-1} \circ F_{|s'|},
    \end{align*}
    which resembles the transport map between ``densities'' $|s'|$ and $|\phi'|$. To generalize this formula for the cases where equality does not hold, we need to make sure that the composition on the right is still well defined, so it is necessary to normalize $|s'|$ and $|\phi'|$ to unit norm 1. As a consequence we use formula $\eqref{eq: g_explicit}$.  
    
   \begin{lemma}[Informal statement -  Proof in Section \ref{app: d_T}]\label{lem: extended graph} 
        Let $s,\phi: [0,1]\to \mathbb R$. If $s\circ g_1=\phi\circ g_2$ for some functions $g_1,g_2\in \mathcal{G}$, then $d_T(s,\phi) = 0$.
    \end{lemma}

    \begin{remark}\label{rmk: graph of transport}
        Let $s,\phi,g_1,g_2$ be as in the previous lemma. Then, $F_{\frac{|s'|}{\|s'\|_1}}\circ g_1 = F_{\frac{|\phi'|} {\|\phi'\|_1}}\circ g_2$. 
        Now, since the cumulative distribution functions are non decreasing, it can be proven we can choose $h_1,h_2\in \mathcal{G}$ such that $F_{\frac{|s'|}{\|s'\|_1}}\circ h_1 = F_{\frac{|\phi'|} {\|\phi'\|_1}}\circ h_2$ and 
        \begin{align}
            g_\phi^s = F_{\frac{|\phi'|} {\|\phi'\|_1}}^\dagger \circ F_{\frac{|s'|}{\|s'\|_1}} = h_2\circ h_1^\dagger. 
        \end{align}
        For details, see Lemma \ref{lm: reparametrization} in Section \ref{app: d_T}.
        More formally, the graph of $g_\phi^s$ is included in the graph of the curve $(h_1,h_2)$. This remark provides an intuition that the warping functions obtained with dynamic time warping divergences are approximately parametrizations of the graph of transport maps (at least, in the perfect matching case).  
    \end{remark}

    \section{Relation between $d_T$ and CDTW and exact classification without noise}
 
    \begin{theorem}[Informal statement -  Proof in Section \ref{app: classif}]\label{thm: cdtw implies sdb}
        Let $s,\phi:[0,1]\to \R$. Then,
        \begin{equation}
            CDTW(s,\phi) = 0 \Longrightarrow d_T(s,\phi) = 0.
        \end{equation}
    \end{theorem}

    \begin{theorem}[CDTW - $d_T$ Guarantees - Proof in Section \ref{app: classif}]\label{thm: CDTW guarantees}
        In the Classification Problem \ref{def: classification problem} assume that the training set contains at least one sample $s_m^{(c)}$ for each atomic class $\mathbb{S}_{\phi_m^{(c)}}$. Then, the 1-NN classification method with divergence given by CDTW 
        perfectly discriminates between classes. That is, given $s$  an unclassified signal, then its class is recovered as
        \begin{equation} \label{eq: classification formula}
            Class(s) = Class\left(\arg \min_{m,c} CDTW(s,s_m^{(c)})\right),
        \end{equation}
        where $Class((m,c)) = c$.
        If instead of $CDTW$ we use $d_T$, we obtain
        \begin{equation} \label{eq: classification formula d_T}
            Class(s) \in Class\left(\arg \min_{m,c} d_T(s,s_m^{(c)})\right).
        \end{equation}
    \end{theorem}

    The pseudocode to compute the class of an unknown signal is provided in Algorithm \ref{alg: dt}. The theorem above is significant in two key ways. First, it establishes that—under the assumptions in Problem Statement \ref{def: classification problem}—classification can be mathematically achieved using only one training sample per atomic class. In this case, perfect classification is attained using the CDTW method.
    Second, by introducing the $d_T$ divergence, and comparing the computational cost of finding DTW alignments vs. transport maps between functions, we make the continuous CDTW formulation tractable and achieve greater computational efficiency compared to the traditional DTW distance (see also Figure \ref{fig: time} below). 

    \begin{algorithm}[H]
        \caption{1D Classification}
        \label{alg: dt}
        \begin{algorithmic}
        \STATE {\bfseries Inputs:}\\ 
            - Sample signals $\{s_m^c\}$ from different atomic classes (finitely many)\\ 
            - Unclassified signal $s$\\ - Choice of dissimilarity function $d(\cdot,\cdot)$:  $d_T$, DTW, or $L^2$-norm
        \STATE {\bfseries Output:}
        Predicted class $c_0$ of signal $s$
        \STATE $c_0 \gets 0$
        \STATE $d_0 \gets \infty$
        \FOR{$m,c$}
            \IF{$d(s,s_m^c) < d_0$}  
                \STATE $d_0 \gets d(s,s_m^c)$
                \STATE $c_0 \gets c$
            \ENDIF
        \ENDFOR
        \STATE {\bfseries return} $c_0$
        \end{algorithmic}
    \end{algorithm}

    \noindent\textbf{Simulated data example: }
        Figure \ref{fig: simulated_low_regime} presents the classification accuracy as a function of the number of training samples on a synthetic data experiment. In this example, training and test samples were generated according to the problem statement described in Problem \ref{def: classification problem}. In this case, each simulated dataset consists of two classes, where each class contains 1 or 5 atomic subclasses. These atomic subclasses are generated from a template $\phi_m^{(c)}$ and deformed using random increasing functions $g\in\mathcal{G}$ by the formula $\phi_m^{(c)}\circ g$ (see Figure \ref{fig: simulated_low_regime} --top part-- for examples of signals). The class of a test sample $s$ is assigned via the nearest neighbor algorithm stated in expression \eqref{eq: classification formula}. 
        \begin{figure}[ht!]
            \centering
            \includegraphics[width=0.65\linewidth]{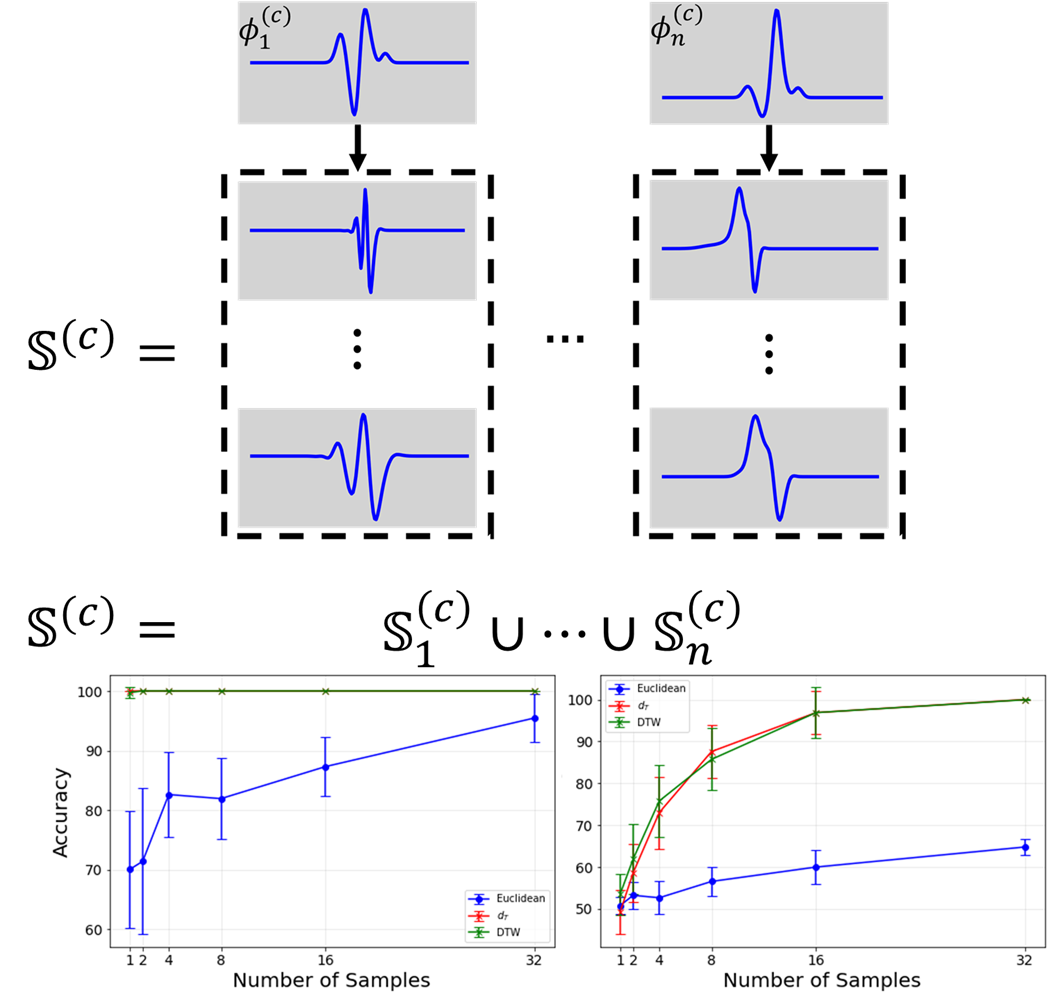}
            \caption{Nearest neighbor classification using different distances on 2 distinct simulated classification problems. The top row contains examples of atomic classes. The bottom row contains the accuracy obtained by different distances as a function of the number of samples in the training set. Left: 1 atomic class per class. Right: 5 atomic classes per class.}
            \label{fig: simulated_low_regime}
        \end{figure}
        
        For classification, we randomly selected 1, 2, 4, 8, 16  and 32 samples per class and applied Equation \eqref{eq: classification formula}. The experiment was repeated 10 times, and we report the average classification accuracies for the 1-nearest neighbor (1NN) classifier using Euclidean, DTW, and $d_T$ metrics. Results show that when there is only one template per class, DTW and $d_T$ achieve $100\%$ accuracy. However, as the number of templates per class increases, DTW and $d_T$ only reach $100\%$ accuracy when at least one sample from each atomic subclass is included in the training set. This phenomenon explains the accuracy trends observed in Figure \ref{fig: simulated_low_regime}.

        Computation time in seconds for the different methods is shown in Figure \ref{fig: time}. It is evident that DTW is several orders of magnitude slower than $d_T$. For real-world datasets, this can become prohibitively expensive. This observation aligns with the fact that computing DTW has a time complexity of $\mathcal{O}(N^2)$ (or, at best, $\mathcal{O}(N^2 / \log\log N)$), whereas computing $d_T$ requires only $\mathcal{O}(N \log N)$ operations. 

        \begin{figure}[ht!]
            \centering
            \includegraphics[width=0.5\linewidth]{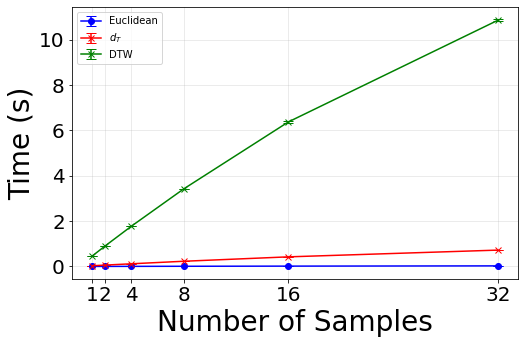}
            \caption{Time requirements for $d_T$, DTW, and Euclidean 1-NN classification with $1,2,4,8,16,32$ samples per class in the training set. Each signal is discretized using 150 evaluation points. The full dataset consists of 250 signals.}
            \label{fig: time} 
        \end{figure}

    \section{Experiments}\label{sec: experiments}
        In this section we describe three experiments performed on data extracted from the UCR time series classification archive \cite{dau2019ucr} in order to demonstrate the effectiveness of the divergence proposed above on real data, as compared to the standard Euclidean and DTW distances. A tutorial on how to use the $d_T$ divergence is included in the Python Library \texttt{PyTransKit} \cite{tutorial}. Three experiments are proposed in order to 1) demonstrate whether the newly proposed method has real world applicability, 2) explore its relation to the traditional DTW distance, and 3) show its performance as a function of the number of training samples. \\

        \noindent\textbf{Experiment 1: } To determine whether there exist real world datasets where the proposed approach can be effective, as compared to existing distances, we evaluate the performance of 1NN classification method using three different distance metrics\footnote{We use the word `metric' in the general sense of the word, although $d_T$ and $DTW$ are not metrics in the mathematical sense but just dissimilarity measures.}: Euclidean, DTW and $d_T$.   These metrics were evaluated on time series datasets from the UCR time series classification archive \cite{dau2019ucr} using their original train and test splits. Table \ref{tab: small_table_details} reports the results from 15 selected datasets where $d_T$ performed well in comparison the other metrics, showing that its implicit model is well adjusted for some non-ideal scenarios. As expected, there is no metric that outperforms the other two in every dataset.  
        In Figure \ref{fig: datasets} we visualize one instance of each class template per dataset. 
        
        \input{small_table_details.tex}
        \begin{figure}[ht!]
            \centering
            \includegraphics[width=0.78\linewidth]{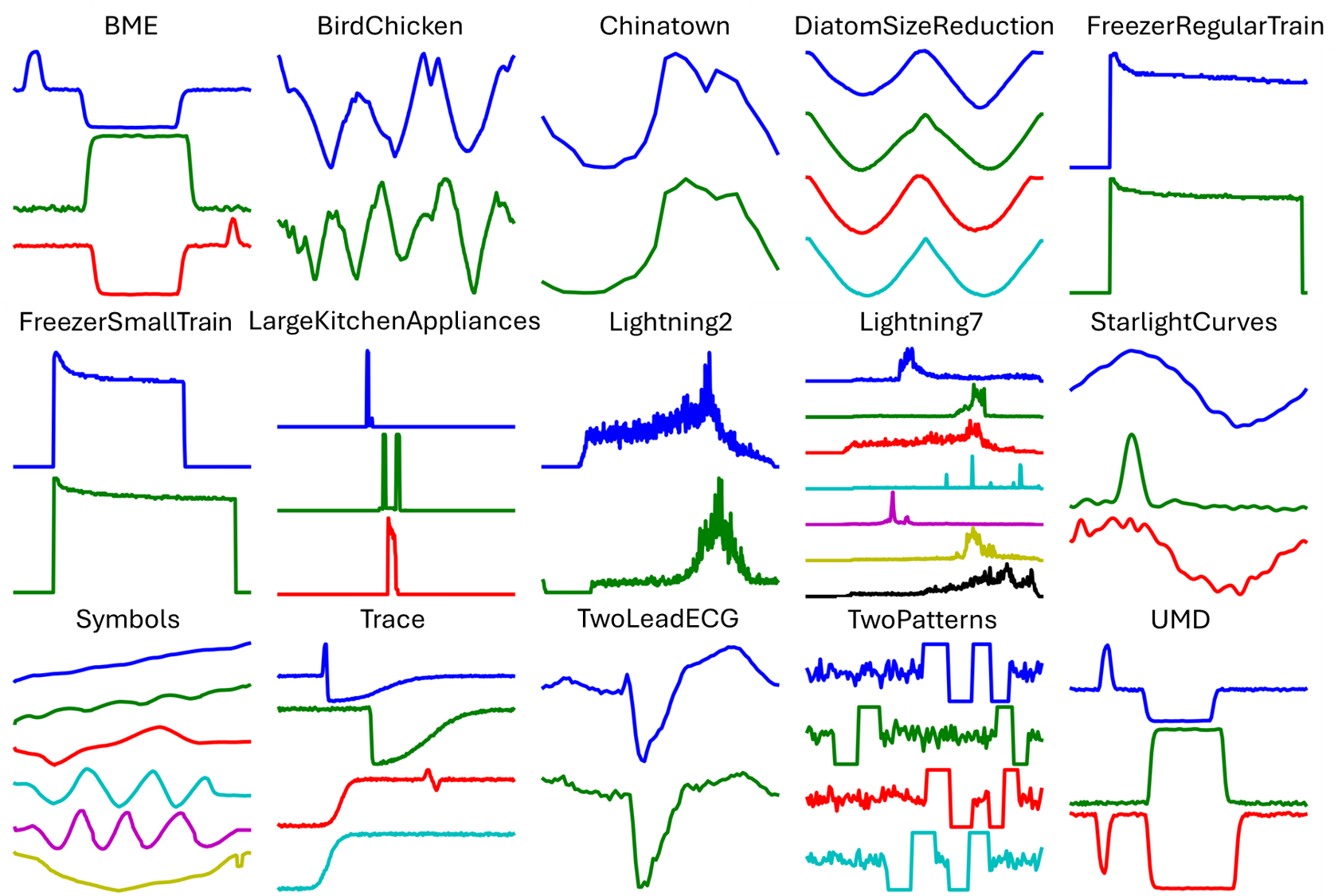}
            \caption{Examples of signals within some datasets from the UCR time series classification archive \cite{dau2019ucr}. Each color represents a sample from a different class. }
            \label{fig: datasets} 
        \end{figure}

        \noindent\textbf{Experiment 2:}  In Figure \ref{fig: correlation}, we demonstrate that the equivalence between the class models of $d_T$ and DTW remains consistent even for non-ideal datasets. The accuracy of both metrics shows a linear correlation in  128 datasets in the UCR archive with a Pearson's correlation coefficient $\rho := 0.87$. Whereas it is not true that $d_T$ outperforms $DTW$ in each dataset, when it does perform similarly, the main advantage is that it is very fast to compute.\\

         \begin{figure}[ht!]
            \centering
            \includegraphics[width=0.65\linewidth]{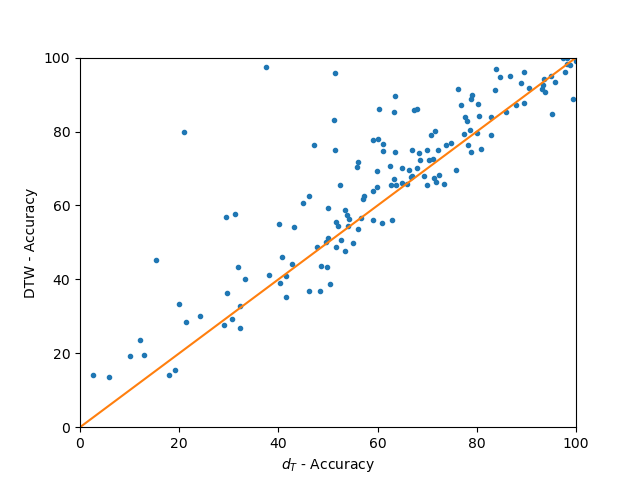}
            \caption{Each point of the plot corresponds to a single dataset for which a 1NN-Classification algorithm has been run using $d_T$ divergence and DTW divergence. The coordinates of each point are the respective accuracies for $d_T$ and DTW. }
            \label{fig: correlation} 
        \end{figure}

        \noindent\textbf{Experiment 3:} 
        For the datasets in Table \ref{tab: small_table_details}, we provide a more thorough analysis of the behavior of the metrics in low-sample regimes. We run the same methods using $1, 2, 4$ and $8$ training samples from each class. This experiment is repeated 30 times, and we report the average accuracies as a function of the number of training samples per class in Figure \ref{fig: low_regime}.

        The conclusions from this experiment are twofold. First, we demonstrate that the implicit class models, DTW and $d_T$, capture the geometry of the classes more effectively than the Euclidean metric. Even with only a few samples per class, we observe high classification accuracies and a rapid convergence toward the asymptotic performance observed with the original train-test splits. Second, we find that the speed advantage of $d_T$ over DTW is consistent: the correlation between the two metrics remains stable regardless of the number of training samples.

        \begin{figure}[ht!]
            \centering
            \includegraphics[width=\linewidth]{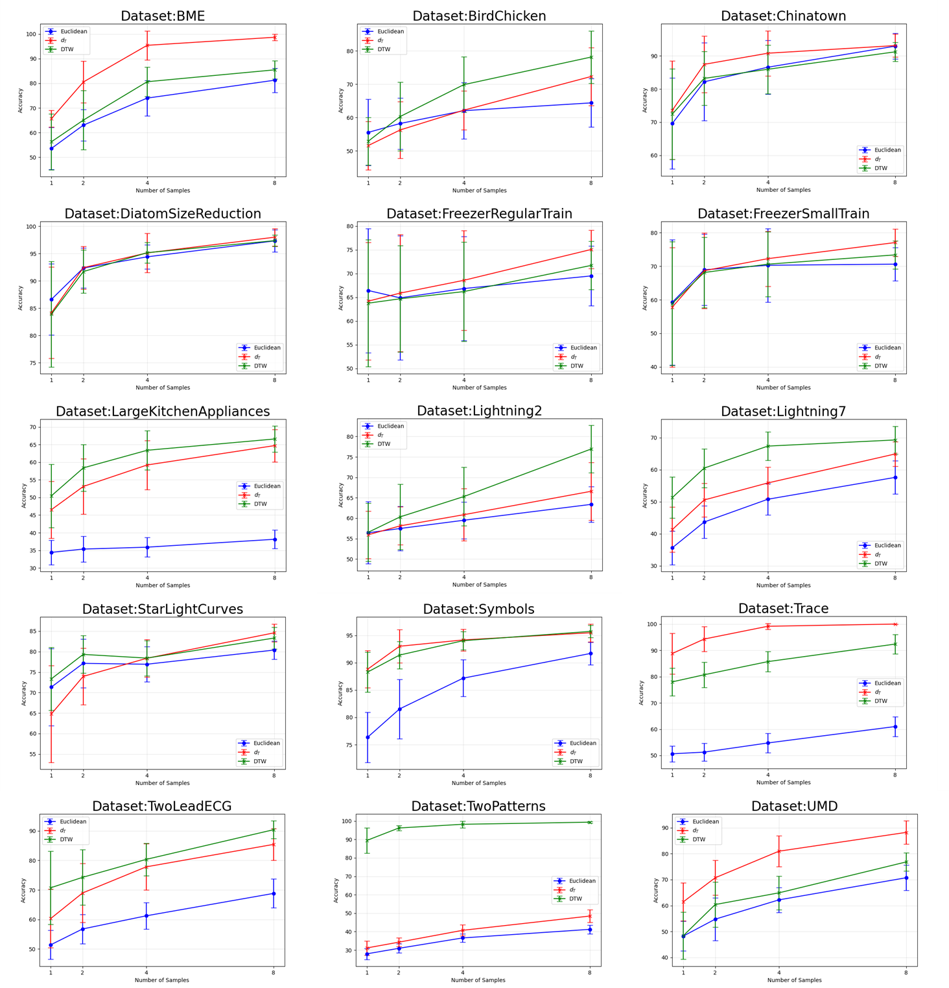}
            \caption{Classification accuracy of nearest neighbor classifiers using different distances.}
            \label{fig: low_regime} 
        \end{figure}

\section{Formal Statements and Proofs}\label{sec: proofs}

\subsection{Proposed Model Class}
\label{app: model class}

    \begin{definition}\label{def: finite flats}
        We say that a function $g$ has \textit{finite flat regions} if the set of zeros of $g'$, $\text{zeros}(g'):= \{x\in Dom(g) : \  g'(x) = 0\}$, has a finite number of connected components. We call each connected component of the set $\text{zeros}(g')$ a flat region of $g$.
    \end{definition}
    Considering functions $g$ under Definition \ref{def: finite flats} allows us to avoid pathological cases, like the Devil's staircase function, whose derivative is zero almost everywhere, but the function is not constant.
    \begin{definition}\label{def: C1 piecewise}
        We say that a function $g:[0,1] \to \mathbb{R}$ is piecewise--$C^1$ if 
        \begin{enumerate}
            \item $g$ is continuous on $[0,1]$.
            \item $g$ is differentiable except at a finite number of points $\{t_1,t_2,\dots,t_n\}$.
            \item $g'$ is continuous in each subinterval $(t_i,t_{i+1})$.
        \end{enumerate}
    \end{definition}
    
    \begin{definition}\label{def: C1_FF}
        We define the set $C^1_{FF}([0,1])$ as the set of all functions $g:[0,1]\to \mathbb{R}$  satisfying Definitions \ref{def: finite flats} and \ref{def: C1 piecewise} simultaneously.
    \end{definition}

    \begin{definition}\label{def: G}
        We define the set of functions $\mathcal{G}$ as 
        \begin{align*}
            \mathcal{G} = \bigg\{& g \in C^1_{FF}([0,1])\mid \quad g \text{ is non-decreasing,} \text{ with } g(0)=0 \text{ and } g(1) = 1 \bigg\}.
        \end{align*}
    \end{definition}
    The definition of $\mathcal{G}$ is inspired by providing a continuous version of the functions $\overline{g}_1, \overline{g}_2$ in the classical definition of the DTW method (see Definition \ref{def: DTW}). 

    \begin{remark}
        If $g\in \mathcal{G}$ is such that it is strictly increasing, then it is invertible, and $g^{-1}\in \mathcal{G}$.
        In general, given $g\in \mathcal{G}$, then 
        \begin{equation}\label{eq: g g dagger = id}
         g\circ g^\dagger = Id   .
        \end{equation}
        This holds from the fact that $g^\dagger$ is strictly right-increasing and the Proposition \ref{prop: properties_gen_inv} on generalized inverses in Appendix \ref{app: gen inv}. However, $g^\dagger$ does not necessarily belongs to $\mathcal{G}$ (for example, $g^\dagger$ will have jumps if $g$ has flat regions).
    \end{remark}

    \begin{definition}\label{def: hyp H}
        For simplicity, we will define hypothesis \textbf{H} as 
        \begin{equation}\label{hip; H}
            s,\phi \in C^1_{FF}([0,1]), 
            \text{ and } \mathcal{G} \text{ is as in Definition \ref{def: G}}. 
        \end{equation}
        This hypothesis is assumed for every result in this article.        
    \end{definition}

\subsection{Proposed Model Class: Equivalence Relation}\label{app: Equivalence relation}

    \begin{proof}[Proof of Proposition \ref{prop: equivalence relation} under hypothesis \textbf{H} (Definition \ref{def: hyp H})]\,
    
    \noindent    We need to verify reflexivity, symmetry and transitivity of the relation
        $s\sim \phi$ if and only if  $s\circ g_1 = \phi \circ g_2$ for some $g_1,g_2\in \mathcal{G}.$

        \textit{Reflexivity and Symmetry:} $s\sim s$ because the identity map belongs to $\mathcal{G}$. Also, if $s\sim \phi$, then there exist $g_1,g_2\in \mathcal{G}$ such that $s\circ g_1=\phi\circ g_2$. Thus, trivially it holds that $\phi\sim s$.
        
       \textit{Transitivity:} If $s\sim \phi$ and $\phi\sim r$, there exists functions $g_1,g_2,f_1,f_2\in \mathcal{G}$ such that $s\circ g_1=\phi\circ g_2$ and $\phi\circ f_1=r\circ f_2$. Let us show that there exists functions $h_1,h_2\in\mathcal{G}$ such that 
        \begin{equation}\label{eq: aux comp}
            g_2\circ h_1=f_1\circ h_2.
        \end{equation}
        The simplest case is when $f_1$ (or, equivalently $g_2$) is strictly increasing and so invertible: We can consider $h_1(t)=t$, $h_2(t)=f_1^{-1}(g_2(t))$.

        In Figure \ref{fig: equivalence_parametrization} we provide a representation of a procedure to find the matching functions $h_1$, $h_2$ in general. For simplicity,  in Figure \ref{fig: equivalence_parametrization} we illustrate $g_2$ and $f_1$ as piecewise linear functions, though this is not necessary. The procedure is to start traversing the domain of $g_2$ linearly with $h_1(t) = t$ and take $h_2(t) = f_1^{-1}\circ g_2(h_1(t))$ while it is possible. If we reach a flat region of $f_1$ we need to switch and traverse the domain of $f_1$ linearly while we keep $h_1$ constant. The actual formulas for $h_1$, $h_2$ are provided below. We need to define a few things before that. 
        
        Let $r_i \in Dom(f_1)$ be the initial point of the $i$-th constant interval of $f_1$. Let $s_i \in Dom(g_2)$ be the first value where $g_2(s_i) = f_1(r_i)$. Let $\Delta_i$ be the length of the $i$-th flat region of $f_1$.

                Let $t_i := s_i + \sum_{j<i} \Delta_j$ and $N(t) := \sum_{i: \, t_i<t} \Delta_i$ (where $t_1=s_1$ and $N(t_1)=0$). We define:
        \begin{align*}
            &h_1(t) = 
            \begin{cases}
                s_i \qquad &\text{ if } t\in [t_i,t_i+\Delta_i)\\
                t - N(t) &\text{ otherwise}
            \end{cases}\\ 
            &h_2(t) = 
            \begin{cases}
                t - t_i +r_i \qquad &\text{ if } t\in [t_i,t_i+\Delta_i)\\
                f_1^\dagger( g_2(h_1(t)))  &\text{ otherwise}
            \end{cases}
        \end{align*}  
        See Figure \ref{fig: equivalence_parametrization}. 
        \begin{figure}[ht!]
            \centering
            \includegraphics[width=\linewidth]{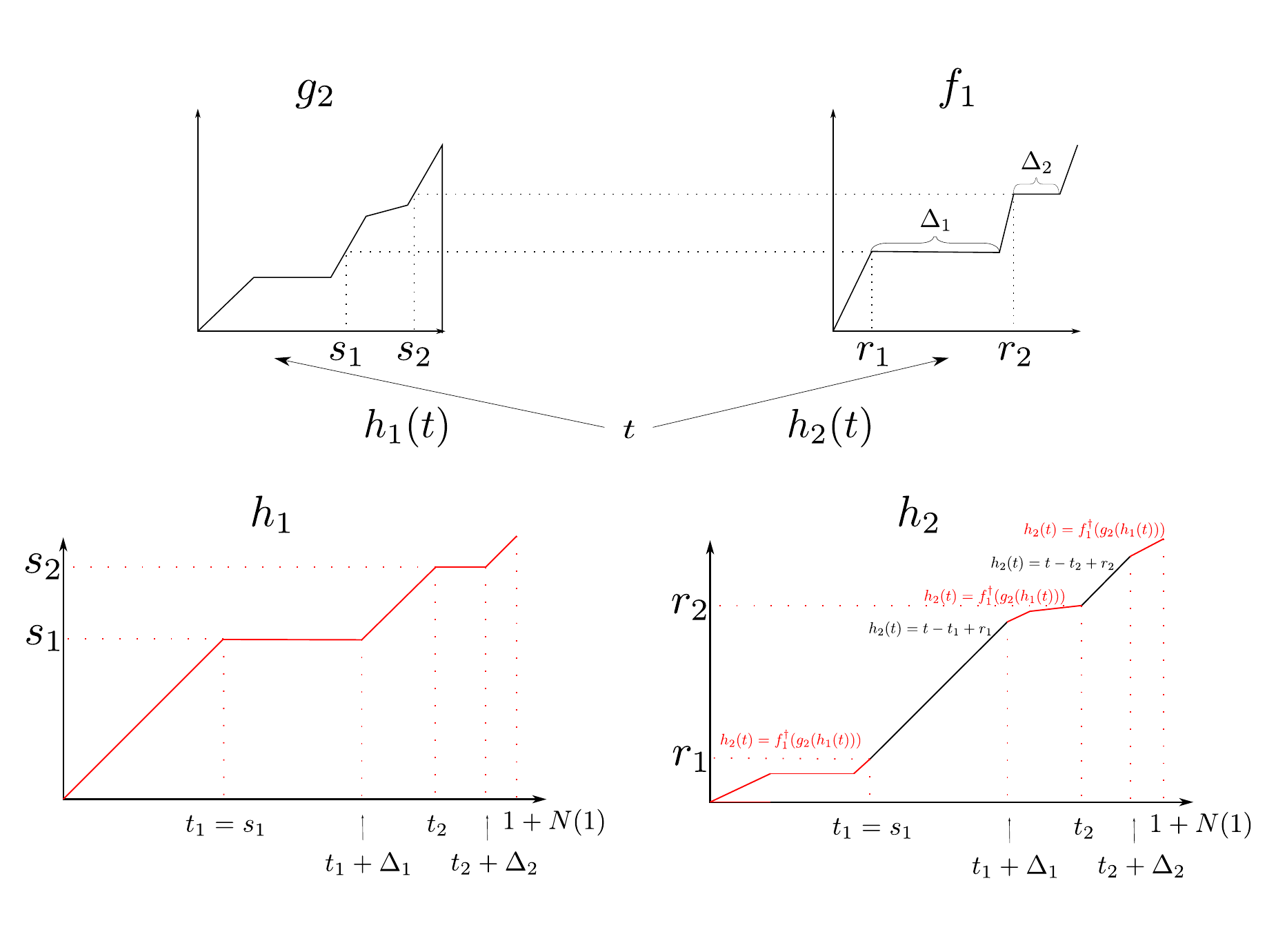}
            \caption{Depicture of the functions $h_1$, $h_2$ (before normalization) used for proving the transitivity property. Let $t$ represent time. We move along the graph of $g_2$ as time progresses according to the rule $h_1(t)$. Respectively, we traverse the graph of $f_1$ over time according to $h_2(t)$.  We start at $t=0$ and as $t$ goes on, we end at the same time. The functions $h_1,h_2$ must be chosen to be non-decreasing, of class $C^1_{FF}$, and so that the value of $g_2$ at time $h_1(t)$ coincides with the value of $f_1$ at time $h_2(t)$. To achieve this, we identify the starting points $r_i$ of the flat regions in the graph of $f_1$, as well as their lengths ($\Delta_i$), and the corresponding points in the domain of $g_2$ that match the same height achieved by $f_1$ at $r_i$. These are which are finitely many, as $f_1\in \mathcal{G}$. } 
            \label{fig: equivalence_parametrization} 
        \end{figure}
        
        Let us check that \eqref{eq: aux comp} holds true for the above choice of $h_1$ and $h_2$:
        If $t\in [t_i,t_i+\Delta_i)$ for some $i$, then $t-t_i+r_i\in [r_i,r_i+\Delta_i)$, and since 
        $f_1(r_i)=f_1(u)$ $\forall u\in [r_i,r_i+\Delta_i]$, we have
        \begin{equation}
         g_2(h_1(t))=g_2(s_i)=f_1(r_i)=f_1(t-t_i+r_i)=f_1(h_2(t)).  
        \end{equation}
        If $t\not\in [t_i,t_i+\Delta_i)$ for any $i$, then, by using \eqref{eq: g g dagger = id},
        \begin{equation*}
            f_1(h_2(t))=f_1(f_1^\dagger( g_2(h_1(t))))=g_2(h_1(t)).
        \end{equation*}

         To end the proof we will show that $h_1$ and $h_2$ have the desire properties. Let us first check that $h_1$ is continuous. It is a piece-wise linear function, and we only have to show continuity at the points of the form $t_i$, $t_i+\Delta_i$. On the one hand, 
        \begin{align*}
            &\lim_{t\to t_1^-}h_1(t)=t_1-N(t_1)=t_1-0=s_1=h_1(t_1) \\
            &\lim_{t\to t_2^-}h_1(t)=t_2-N(t_2)=t_2-\Delta_1=s_2=h_1(t_2)
        \end{align*}
        and, in general,
        \begin{equation*}
            \lim_{t\to t_i^-}h_1(t)=t_i-N(t_i)=t_i-\sum_{j<i}\Delta_j=s_i=h_1(t_i).
        \end{equation*}
        On the other hand,
        \begin{align*}
            \lim_{t\to (t_i+\Delta_i)^-}h_1(t)&=s_i =t_i-\sum_{j<i}\Delta_j=t_i-\Delta_i-\sum_{j<i+1}\Delta_j\\
            &=t_i+\Delta_i-N(t_i+\Delta_i) = h_1 (t_i+\Delta_i).
        \end{align*}
        Now, let us check that $h_2$ is continuous: 
        We only need to show continuity at the points of the form $t_i$, $t_i+\Delta_i$. Indeed, $f_1^\dagger\circ g_2\circ h_1$ is continuous on $(t_i+\Delta_i,t_{i+1})$ because outside the intervals $(t_i,t_{i}+\Delta_i)$ we have $f_1^\dagger=f_1^{-1}$ and so it is continuous, $g_2$ is continuous since it belongs to the set $\mathcal{G}$ and we have already proven that $h_1$ is continuous.  
        On the one hand,
        \begin{align}
            \lim_{t\to t_i^-}h_2(t)&=\lim_{t\to t_i^-}f_1^\dagger (g_2(h_1(t)))=\lim_{t\to t_i^-}f_1^\dagger (g_2(t-N(t))). \label{eq: lim h_2}
        \end{align}   
        Since  $g_2(t-N(t))<g_2(s_i)$ for $t<t_i$, by definition of the generalized inverse we obtain
        \begin{align*}
            f_1^\dagger (g_2(t-N(t)))&=\inf\{z\mid \, f(z)>g\}\leq \inf \{z\mid \, f_1(z)=g_2(s_i)\}=r_i.
        \end{align*}
        Thus, from \eqref{eq: lim h_2}
        \begin{equation*}
            \lim_{t\to t_i^-}h_2(t)\leq r_i.
        \end{equation*}
        By contradiction, suppose that
        \begin{equation}\label{eq: h_2 contradic 1}
            \lim_{t\to t_i^-}h_2(t)< r_i.
        \end{equation}
        Applying on both sides the function $f_1$, which is continuous, using the identity \eqref{eq: g g dagger = id}, and the continuity of $g_2$, we have 
        \begin{align*}
            f_1(r_i)&\geq f_1\left(\lim_{t\to t_i^-}h_2(t)\right)\\
            &=\lim_{t\to t_i^-}f_1(f_1^\dagger (g_2(t-N(t))))\\
            &=\lim_{t\to t_i^-}g_2(t-N(t))\\
            &=g_2(s_i)=f_1(r_i).
        \end{align*}
        That is,
        \begin{equation}\label{eq: h_2 contradic 2}
            f_1(r_i)= f_1\left(\lim_{t\to t_i^-}h_2(t)\right).
        \end{equation}
        Thus, from $\eqref{eq: h_2 contradic 1}$ and \eqref{eq: h_2 contradic 2} we get a contradiction since by definition of $r_1$ is the minimum value where $f_1$ achieves its $i$-th flat region. Therefore, 
        $$\lim_{t\to t_i^-}h_2(t)=r_i=t_i-t_i+r_i=h_2(t_i).$$
        On the other hand, 
        \begin{align}
            \lim_{t\to (t_i+\Delta_i)^-}h_2(t)&=t_i-\Delta_i+t_i+r_i=r_i+\Delta_i \label{eq: h2 aux 11}
        \end{align}
        and, by definition of the generalized inverse,
        \begin{align}
            h_2(t_i+\Delta_i)&=f_1^\dagger(g_2(h_1(t_i+\Delta_i))) \notag\\
            &=f_1^\dagger(g_2(s_i)) \notag \\
            &=\inf\{z\mid \, f_1(z)>g_2(s_i)\} \notag \\
            &=\sup\{z\mid \, f_1(z)=g_2(s_i)\} \notag \\
            &=r_i+\Delta_i \label{eq: h2 aux 22}
        \end{align}
        where the previous to last equality follows from the fact that $f_1$ is a continuous function, and the last equality follows because $f_1$ is constant on the interval $[r_i,r_i+\Delta_i]$. Therefore, from \eqref{eq: h2 aux 11} and \eqref{eq: h2 aux 22} it follows that $h_2$ is continuous at the points of the form $t_i+\Delta_i$. 

        Finally, notice that $h_1,h_2$ have domain $[0,1+N(1)]$. But as long as they satisfy $g_2 \circ h_1 = f_1\circ h_2$ this will pose no problem since we can rescale them to be defined on $[0,1]$ by taking $ h_i(t(1+N(1))$, $i=1,2$.
        After rescaling, $h_1,h_2\in \mathcal{G}$ since they will satisfy $h_1(0)=0=h_2(0)$ and $h_1(1)=1=h_2(1)$ and, by construction, they have finitely many plateaus and are piece-wise continuously differentiable. See Figure \ref{fig: equivalence_parametrization}.
    \end{proof}

    \begin{remark}\label{rmk: parametrization}
        Let $NF = \left(\bigcup[t_i,t_i+\Delta_i)\right)^c$. From the definition of $h_1$ and Figure \ref{fig: equivalence_parametrization} it is easy to notice that $$Im(h_1) = Im({h_1}|_{NF}) = Dom(g_2) = Dom(f_1^\dagger\circ g_2).$$ 
        Now, using the definition of $h_2$ we obtain that for every $u\in Dom(f_1^\dagger\circ g_2)$, there exists a $t$ such that $h_1(t) = u$ and $h_2(t) = f_1^\dagger\circ g_2(h_1(t)) = f_1^\dagger\circ g_2(u) $.
    \end{remark}
    
\subsection{The CDTW divergence}\label{app: CDTW}

\begin{proof}[Proof of Theorem \ref{thm: cdtw class} under hypothesis \textbf{H} (Definition \ref{def: hyp H})]\,

    Let $s:[0,1]\to \mathbb R$, $\phi:[0,1]\to \mathbb R$ in $C^1_{FF}([0,1])$.
     Let $\gamma^*(t)=(g_1(t),g_2(t))$ be 
     such that          $$0=\int_0^1
        |s(g_1(t))-\phi(g_2(t))|\ |g_1'(t)+g_2'(t)|dt.$$
        Recall that $|g_1'(t)+g_2'(t)|=g_1'(t)+g_2'(t)$ since $g_1',g_2'\geq0$.

        The set of points where $(g_1+g_2)'$ does not exists is finite. Let denote such points $0\leq t_1<\dots<t_{n-1}\leq 1 $ and let $t_0=0, t_n=1$.
        Given $0\leq j\leq n-1$, consider the interval $(t_j,t_{j+1})$ where $g_1+g_2$ is of class $C^1$.
        Since $g_1,g_2\in \mathcal{G}$, the set
        $$Z_j:=\{t\in (t_j,t_{j+1}): \, g_1'(t)=0 \text{ and } g_2'(t)=0\}$$
         has a finite number of connected components. Let us denote them as  $C_1^j,\dots,C_m^j$. Recall that each connected component $C_i^j$ is an interval on $\R$ which could be the degenerate, i.e.,  a singleton.

        Consider an arbitrary point  $t_0\in (t_j,t_{j+1})$. 

        If $t_0\not\in Z_j$, then  $g_1'(t_0)+g_2'(t_0)>0$. By continuity, there exists  $\varepsilon>0$ such that
        $g_1'(t)+g_2'(t)>0$ for all $t\in(t_0-\varepsilon,t_0+\varepsilon)$. Thus, it follows that 
        $s(g_1(t))=\phi(g_2(t))$ for all $t\in(t_0-\varepsilon,t_0+\varepsilon)$.

        If $t_0\in Z_j$, let $C_i^j$ be the connected component containing $t_0$.

        Suppose $C_i^j$ is a singleton, that is, $C_i^j=\{t_0\}$. Then there exists $\varepsilon>0$ such that $(g_1+g_2)'>0$ on $(t_0,t_0+\varepsilon)$ and on $(t_0-\varepsilon,t_0)$. By using this and  also the continuity of the functions $s\circ g_1$ and $\phi\circ g_2$, we have that $s(g_1(t))=\phi(g_2(t))$ for all $t\in(t_0-\varepsilon,t_0+\varepsilon)$.

        Thus, so far we have proven that $s\circ g_1-\phi\circ g_2$ is the constant function zero outside the non-degenerate connected components of $Z_j$.

        Suppose $C_i^j$ is a non-degenerate interval. Then, since $g_1'\equiv 0$ and $g_2'\equiv 0$ on $C_i^j$, we have that $g_1$ and $g_2$ are constant on $C_i^j$. So, $s(g_1(t))-\phi(g_2(t))$ is also constant on $C_i^j$.

        Therefore, we have proven that on $(t_j,t_{j+1})$ the function $s\circ g_1-\phi\circ g_2$ is a continuous functions whose range is discrete. Thus, $s\circ g_1-\phi\circ g_2$ is constant on $(t_j,t_{j+1})$, and this holds for every $0\leq j\leq n-1$. 
        Hence, we see that on $[0,1]$ the function $s\circ g_1-\phi\circ g_2$ is a continuous functions whose range is discrete. Thus,  $s\circ g_1-\phi\circ g_2$ is constant $c$ on $[0,1]$.
        Going back to our hypothesis, we have
        \begin{align*}
             0&=\int_0^1
            |s(g_1(t))-\phi(g_2(t))|\ |g_1'(t)+g_2'(t)|dt=c\left(\int_0^1
            |g_1'(t)+g_2'(t)|dt\right)  .
        \end{align*}
        If $c\not=0$, then $\int_0^1
        |g_1'(t)+g_2'(t)|dt=0$ which implies that $g_1'=0=g_2'$ everywhere except for a finite number of points where they are not defined. Since $g_1$ and $g_2$ are continuous functions on $[0,1]$, we have that they are constant on $[0,1]$ But the functions $g_1,g_2$ can not be constant since $g_1(0)=0=g_2(0)$ and $g_1(1)=1=g_2(1)$. Thus, $c=0$.
    \end{proof}
    \begin{remark}
        The finite flat regions hypothesis is not easy to relax. For example, if we would only ask for a countable number of flat regions, then setting $x(t)$ and $y(t)$ as Cantor's functions it can be proven that $CDTW(s,\phi) =0$ for every $x,\phi$. Thus Theorem \ref{thm: cdtw class} is no longer valid. Also piecewise-$C^1$  function could be constructed that do not satisfy the finite flat regions hypothesis.
    \end{remark}
    
\subsection{Relation between CDT and DTW divergences}\label{app: CDTW-DTW}
        
        The proof of Theorem \ref{thm: CDTW-DTW param space} is inspired by the approximation algorithm in Chapter 4 of the Master's Thesis \cite{klaren2020continuous}. For that, we will work with curves $\gamma:[0,1] \to [0,1]\times[0,1]$ and some clarification are needed. We will equip the square $[0,1]\times[0,1]$ with the 1-norm $\|(x,y)\|_1 = |x| + |y|$. For any curve $\gamma:[0,1]\to [0,1]\times[0,1]$ we consider its arc-length respect to this norm as the usual limit 
        \begin{equation*}
            Length(\gamma) = \lim_{M\to \infty} \sum \left\|\gamma\left(\frac{i+1}{M}\right) - \gamma\left(\frac{i}{M}\right)\right\|_1,
        \end{equation*}
        which coincides with the usual formula
        \begin{equation*}
            Length(\gamma) = \int_0^1 \|\gamma'(t)\| dt
        \end{equation*}
        when $\gamma$ is sufficiently smooth. We say that the value $\|\gamma'(t)\|_1$ is the speed of the curve at time $t$.
        It is worth noticing that if the coordinate functions are non-decreasing and $\gamma(0) = (0,0)$, $\gamma(1) =(1,1)$, then $Length(\gamma) = \|(1,1)-(0,0)\|_1 = 2$. We say that $\gamma$ is parametrized by constant speed if there exists $c>0$ such that $\|\gamma'(t)\|_1 = c$ for all differentiable points $t$ of $\gamma$, and this constant $c$ should be equal to its length if the domain is $[0,1]$, i.e., $c=Length(\gamma)$.

    \begin{proof}[Proof of Theorem \ref{thm: CDTW-DTW param space}  under hypothesis \textbf{H} (Definition \ref{def: hyp H})]\, 

   \noindent  Let $s:[0,1]\to \mathbb R$, $\phi:[0,1]\to \mathbb R$ in $C^1_{FF}([0,1])$.
            Let $\{u_j\}_{j=0}^N$ be a uniform partition of $[0,1]$ with step size $\Delta$ to be determined, and consider discretizations $\overline{s}$, $\overline{\phi}$ of $s$, $\phi$ over the partition $\{u_j\}_{j=0}^N$ as Section \ref{sec: dtw - cdtw}.                      
        We define the lattice $\{(u_i,u_j)\}_{i,j =0}^N$ in $[0,1]\times [0,1]$. We say that the square with vertices $\{(u_{i-1},u_{j-1}),(u_{i-1},u_j),(u_i,u_{j-1}),(u_i,u_j)\}$ is the cell $C_{ij}$.
        The diameter of each cell is computed with respect to the norm $\| \cdot\|_1$ as
        \begin{equation*}
        \|(u_{i},u_{j})-(u_{i-1},u_{j-1})\|_1=    |u_j -u_{j-1}| + |u_i - u_{i-1}| =  2\Delta   .
        \end{equation*}

        \textbf{Step 1:} First, let us assume that there exists $\gamma^*:[0,1] \to [0,1]\times [0,1]$ be an optimal continuous warping path for  $CDTW(s,\phi)$. 
        We will see how to relax this assumption later. (See Step 4 at the end of the proof.)

        The path $\gamma^*(t)=(x^*(t),y^*(t))$ visits a finite number of cells $K$, with $N\leq K \leq 2N$, in the partition of $[0,1]\times [0,1]$.
        For every cell $C_{ij}$ that is visited by $\gamma^*$ consider the bottom-left vertex $(u_{i-1},u_{j-1})$. That is, varying $i,j \in \{1,\dots,N\}$, consider the set of vertices
        of the form $(u_{i-1},u_{j-1})\in [0,1]\times[0,1]$ such that there exists $t\in [0,1]$  satisfying $\gamma^*(t) \in C_{ij}$. Add to such set the point $\gamma^*(1)=(1,1)$ (i.e., the top-right vertex of $[0,1]\times [0,1]$).
        Order the points in such set by $\leq_{xy}$ defined by $( x_1, y_1)\leq_{xy}(x_2,  y_2)$ if and only if $ x_1\leq  x_2$ and $ y_1\leq  y_2$, and rename them as
        $\{v_k\}_{k=0}^{K}$. Notice that $v_0=(0,0)$, $v_K=(1, 1)$, and that all the points in this set of vertices are comparable according to the partial order $\leq_{x,y}$ since $x^*,y^*$ are non-decreasing. See Figure \ref{fig: paths}.
        
        For each $1\leq k\leq K$, we define the $k$-th edge $E_k$ as the segment in $[0,1]\times [0,1]$ with extreme points $v_{k-1}$ and $ v_{k}$. Notice that 
        $E_k$ will be either a horizontal or vertical side of a cell or its diagonal.

        We define $\gamma_{e}:[0,1] \to [0,1]\times [0,1]$, $\gamma_e(t) = ({x}_e(t),{y}_e(t))$, as the polygonal curve generated by the sequence  of vertices $\{v_k\}_{k=0}^{K}$ in $[0,1]\times [0,1]$ \textit{parametrized by constant speed} (where it is differentiable) \textit{with respect to} $\|\cdot\|_1$\textit{-norm in} $[0,1]\times [0,1]$ (as stated above). That is,
        \begin{align}\label{eq: len}
            \|\gamma_e'(t)\|_1 &:= |x'_e(t)| + |y'_e(t)| =  2 \quad \forall t\in[0,1], \notag \\
            &\text{ and } \int_0^1\|\gamma_e'(t)\|_1 \, dt=2.
        \end{align}        
        It won't necessarily be optimal for $CDTW(s,\phi)$, but we will use it to ``estimate'' the cost given by $\gamma^*$. 
        For each $0\leq k\leq K$, let us define $t_k\in [0,1]$ as the time such that $\gamma_{e}$ visits a vertex $v_k$. Thus, $\gamma_e(t_k) = ({x}_e(t_k),{y}_e(t_k)) =v_k$. It holds that $\gamma_e$ is an admissible warping path for CDTW optimization problem \eqref{eq: cdtw12} (see Definition \ref{def: CDTW}).

        Let $\gamma_{e}^\perp:[0,1]\to[0,1]\times[0,1]$ be the reparametrization of $\gamma^*$ defined as follows:
        For each time $t\in [0,1]$, the point $\gamma_{e}^\perp(t)$ is given by  the projection of $\gamma_{e}(t)$ onto the set $\gamma^*([0,1])$ along the direction vector $(1,-1)$. It holds that $\gamma_e^\perp$ is an admissible path in the sense of Definition \ref{def: CDTW}. Note that $\gamma^*$ and $\gamma_e^\perp$ are different parametrizations of the same curve. See Figure \ref{fig: paths}.
        
        \begin{figure}
            \centering
            \includegraphics[width=1\linewidth]{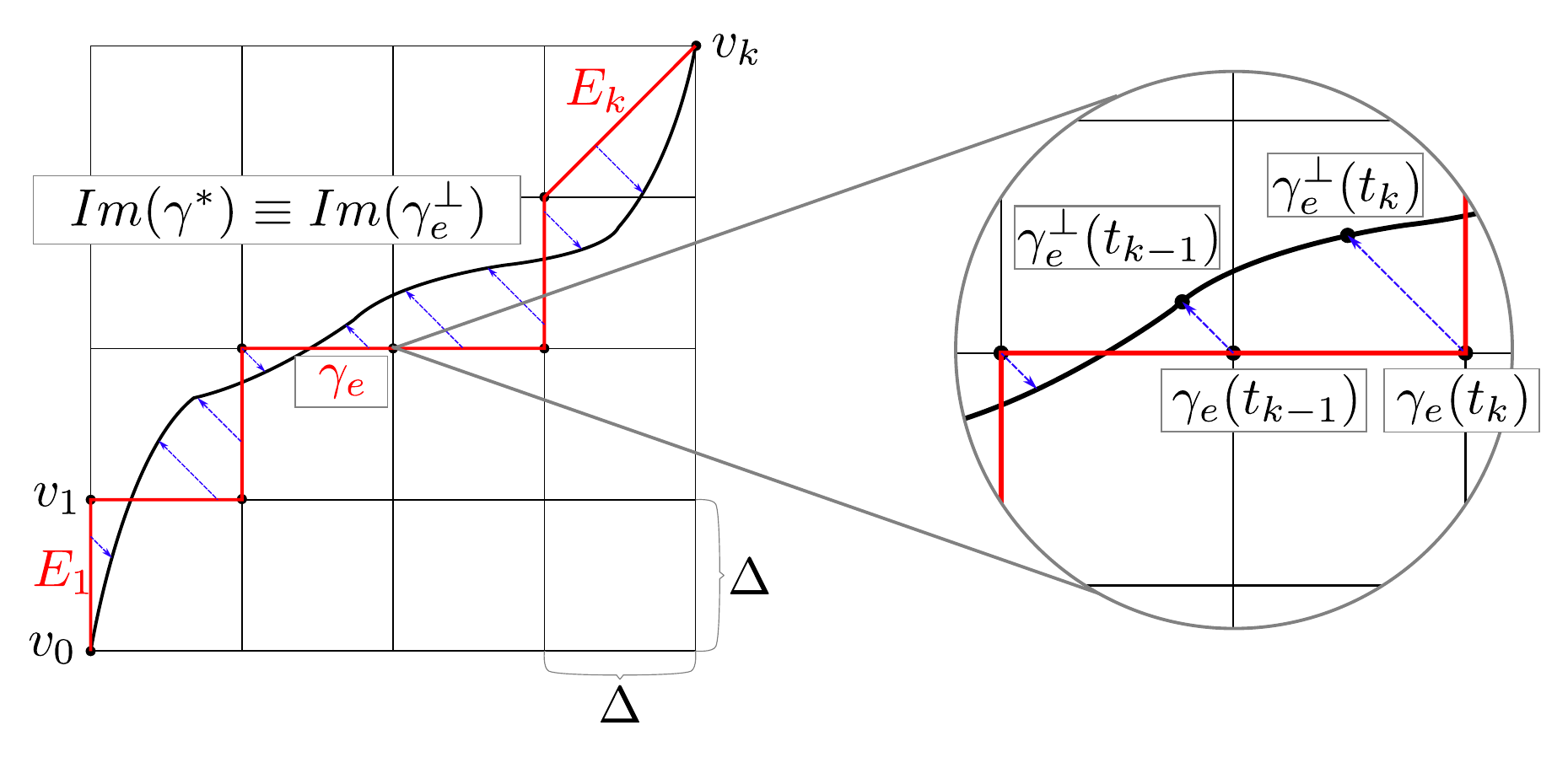}
            \caption{
             Discretization of $[0,1]\times[0,1]$ into cells of size $\Delta\times \Delta$.
            Visualization of the optimal path $\gamma^*$ (black curve), the points $v_0$, $v_1$, $\dots$, $v_K$ obtained as $\gamma^*$ crosses one cell, and its union $\gamma_e$ (red polygonal curve). The curve $\gamma_e^\perp$ is obtained by projecting $\gamma_e$ onto $\gamma^*$ along the line with direction $(1,-1)$ (blue arrows). As a result, the image of $\gamma_e^\perp$ coincides with the image of $\gamma^*$ (black curve). }
            \label{fig: paths}
        \end{figure}

        Now, consider the discretizations $\overline{s}$, $\overline{\phi}$ of $s$, $\phi$ over the partition $\{u_j\}_{j=0}^N$.
        Then, the discrete warping path defined by $$\overline{\gamma}_e (k) := (x(k),y(k)) =(x_e(t_k),y_e(t_k))/\Delta$$ 
        is an admissible path for $DTW_\omega(\overline{s},\overline{\phi})$. Moreover, it satisfies 
        \begin{align} \label{eq: dicrete_cont_h}
            |\overline{s}(x(k))-\overline{\phi}(y(k))| = |s(x_e(t_k))-\phi(y_e(t_k))|.
        \end{align}
        The weight function for $\overline{\gamma}_e$, as defined in \eqref{eq: weight}, takes the form        
        \begin{align} \label{eq: discrete_cont_weights}
            \omega_k&=
                          |x_e({t_k})-x_e(t_{k-1})| + |y_e({t_{k}})-y_e{(t_{k-1})}|,
        \end{align}
        {which can have values $\Delta$ or $2\Delta$.}

        \begin{claim}(See proof below.)\label{claim: weights} For each $1\leq k\leq K$,
                    $\int_{t_{k-1}}^{t_k}\|\gamma_e'(t)\|_1dt=\omega_k$.
        \end{claim}

        \begin{claim}(See proof below.)\label{claim: norm1 pi}  
            $\|(\gamma_e^\perp)'(t)\|_1=\|\gamma_e'(t)\|_1$.
        \end{claim}

        Using Claims \ref{claim: weights} and \ref{claim: norm1 pi}, and the fact that $\gamma_e^\perp$  is a reparametrization of $\gamma^*$, we compute 
        \begin{align}
            CDTW(s,\phi)&=\int_0^1 |s(x^*(t))-\phi(y^*(t))| \|(\gamma^*)'(t)\|_1 dt
            \notag \\
            &=\int_0^1 |s(x_e^\perp(t))-\phi(y_e^\perp(t))|\|(\gamma_e^\perp)'(t)\|_1 dt\notag%
            \\
            &=\sum_{k=1}^K\int_{t_{k-1}}^{t_k} |s(x_e^\perp(t))-\phi(y_e^\perp(t))|\|(\gamma_e^\perp)'(t)\|_1 dt \notag\\
            &=\sum_{k=1}^K\int_{t_{k-1}}^{t_k} |s(x_e^\perp(t))-\phi(y_e^\perp(t))|\|\gamma_e'(t)\|_1 dt. \label{eq: aux cdtw}
        \end{align}    
        Using \eqref{eq: dicrete_cont_h}, \eqref{eq: discrete_cont_weights}, and Claim \ref{claim: weights}, we obtain
        \begin{align}
            DTW_\omega(\overline{s},\overline{\phi})&\leq \sum_{k=1}^K  |s(x_e(t_k))-\phi(y_e(t_k))|\omega_k \notag\\
            &=\sum_{k=1}^K |s(x_e(t_k))-\phi(y_e(t_k))| \int_{t_{k-1}}^{t_k} \|\gamma_e'(t)\|_1 dt \notag\\ 
            &=\sum_{k=1}^K\int_{t_{k-1}}^{t_k}|s(x_e(t_k))-\phi(y_e(t_k))| \|\gamma_e'(t)\|_1 dt. \label{eq: aux dtw}
        \end{align}

        \begin{claim}\label{claim: h_tilde_bound}(See proof below.)
            Consider the cost function $h(\gamma(t))$
            as in the Definition \ref{def: CDTW}.
            For all $t\in (t_{k-1},t_k)$, we have $\|(\gamma_e^\perp)(t)-\gamma_e(t_k)\|_1\leq 3{\Delta}.$
            As a consequence, for all $t\in (t_{k-1},t_k)$, we have      $\left|h(\gamma_e^\perp(t))- h(\gamma_e(t_k))\right|\leq \max\{\|s'\|_\infty,\|\phi'\|_\infty\} \cdot 3\Delta$. 
        \end{claim}
    
        Let us denote $\widetilde\Delta:=\max\{\|s'\|_\infty,\|\phi'\|_\infty\} \cdot 3\Delta$.
        Using \eqref{eq: len}, \eqref{eq: aux cdtw}, \eqref{eq: aux dtw}, and Claim \ref{claim: h_tilde_bound}, we get 
        \begin{align}
            DTW_\omega(\overline s, \overline \phi)-CDTW(s,\phi)&\leq \sum_{k=1}^K\int_{t_{k-1}}^{t_k} \left( h(\gamma_e^\perp(t))- h(\gamma_e(t_k))\right)\|\gamma_e'(t)\|_1 dt \notag \\ 
            &\leq \sum_{k=1}^K\int_{t_{k-1}}^{t_k} \left| h(\gamma_e^\perp(t))- h(\gamma_e(t_k))\right|\|\gamma_e'(t)\|_1 dt \notag \\
            &\leq \widetilde{\Delta} \sum_{k=1}^K \int_{t_{k-1}}^{t_k} \|\gamma_e'(t)\|_1 dt \notag \\
            &= \widetilde{\Delta} \int_{0}^{1} \|\gamma_e'(t)\|_1 dt \notag \\
            &= 2\widetilde{\Delta} .\label{eq: first inequality}
        \end{align}
        
       \textbf{Step 2:} We still need to bound the absolute value of the left-hand side of the previous inequality. To achieve this, let $\overline \gamma_d$ be an optimal discrete warping path between $\overline s$ and $\overline \phi$ 
        for  the problem given in \eqref{eq: DTW problem w}. 
        The range of $\overline \gamma_d$ is a set of vertices $\{N\cdot v_0,N\cdot v_1, \dots, N\cdot v_{J-1}, N\cdot  v_{J}\}$, for some $J \leq 2 N$, and  $v_j\in [0,1]\times [0,1]$ with $v_0=(0,0)$ and $v_J=(1,1)$. 
                
        Let $\gamma_d(t)=(x_d(t),y_d(t))$ denote the polygonal continuous path form $[0,1]$ to $[0,1]\times[0,1]$ that connects the vertices $\{v_j\}_{j=0}^J$ in order and is parametrized by constant-speed. 
        
        On the one hand, 
        \begin{equation*}
            CDTW(s,\phi)\leq \int_0^1{h}(\gamma_d(t))\|\gamma_d'(t)\|_1dt. 
        \end{equation*}
        
        On the other hand, similarly as before, for each $1\leq j\leq J$ let us denote $t_j\in [0,1]$ the time such that $\gamma_{d}$ visits the vertex $v_j$, that is,  $\gamma_d(t_j) = ({x}_d(t_j),{y}_d(t_j)) 
         =v_j$.  Thus, by using Claim \ref{claim: weights} applied to $\gamma_d$ we can write
        \begin{align*}
            DTW_\omega(\overline{s},\overline{\phi})&=\sum_{j=1}^J|s(x_d(t_j))-\phi(y_d(t_j))|\omega_j=\sum_{j=1}^J h(\gamma_d(t_j))\int_{t_{j-1}}^{t_j}\|\gamma_d'(t)\|_1dt.
        \end{align*}
        
        Therefore,
        \begin{align}
            CDTW(s,\phi)-DTW_\omega(\overline{s},\overline{\phi})
            &\leq \sum_{j=1}^J\int_{t_{j-1}}^{t_j}|{h}(\gamma_d(t))-{h}(\gamma_d(t_j))|\|\gamma_d'(t)\|_1dt\notag \\
            &\leq \widetilde\Delta\sum_{j=1}^J\int_{t_{j-1}}^{t_j}\|\gamma_d'(t)\|_1dt\notag \\
            &=\widetilde\Delta\int_0^1\|\gamma_d'(t)\|_1dt \notag \\
            &= 2\widetilde{\Delta} .  \label{eq: second inequality}
        \end{align}

        \textbf{Step 3:} In conclusion, from \eqref{eq: first inequality} and \eqref{eq: second inequality}, given $\varepsilon>0$ by choosing step-size
        $$\Delta<\delta:=\frac{\varepsilon}{6 \max\{\|s'\|_\infty,\|\phi'\|_\infty\}}$$ we obtain     
        \begin{equation}\label{eq: final}
                |CDTW(s,\phi)-DTW_\omega(\overline{s}, \overline{\phi})|<\varepsilon.
        \end{equation}        

        \textbf{Step 4:} Finally, if there is no optimal solution $\gamma^*$ for $CDTW(s,\phi)$, given $\varepsilon>0$ let us consider an admissible continuous warping path $\gamma_{{\varepsilon}}$ such that
        \begin{align} 
            &\int_0^1
            h(\gamma_{{\varepsilon}}(t))\|(\gamma_{{\varepsilon}})'(t)\|_1dt -CDTW(s,\phi) \notag \\
            &=\left|CDTW(s,\phi)-\int_0^1
            h(\gamma_{{\varepsilon}}(t))\|(\gamma_{{\varepsilon}})'(t)\|_1dt\right|
            <\frac{\varepsilon}{2}.\label{eq: approx minimizer 1}
        \end{align}
        By repeating the proof up to \eqref{eq: first inequality} with $\gamma_{\varepsilon}$ playing the role of $\gamma^*$, we obtain that there exists a $\delta_1$ independent of $\gamma_{{\varepsilon}}$ such that if $s,\phi$ are discretized with step size $\Delta <\delta_1$ then, analogously to \eqref{eq: first inequality},  
        \begin{equation} 
            DTW_\omega(\overline{s},\overline{\phi})  - \int_0^1
            h(\gamma_{{\varepsilon}}(t))\|(\gamma_{{\varepsilon}})'(t)\|_1dt< \frac{\varepsilon}{2}. \label{eq: approx minimizer 2}
        \end{equation}
        Using \eqref{eq: approx minimizer 1} and \eqref{eq: approx minimizer 2}, 
        \begin{equation*}
                DTW_\omega(\overline{s}, \overline{\phi}) - CDTW(s,\phi)<\varepsilon.
        \end{equation*}    

        Notice that for \eqref{eq: second inequality} the existence of a minimizer $\gamma^*$ was not needed. In conclusion, for every $\varepsilon >0$, choosing a step-size
        \begin{equation*}
            \Delta<\delta:=\min\left\{\frac{\varepsilon}{6 \max\{\|s'\|_\infty,\|\phi'\|_\infty\}}, \delta_1\right\}
        \end{equation*}
        we obtain     
        \begin{equation*}
            |CDTW(s,\phi)-DTW_\omega(\overline{s}, \overline{\phi})|<\varepsilon.
        \end{equation*} \ 
    \end{proof}

    \begin{proof}[Proof of Claim \ref{claim: weights}] 
            \begin{align*}
                \int_{t_{k-1}}^{t_k}\|\gamma_e'(t)\|_1dt&=\int_{t_{k-1}}^{t_k} | x_e' (t)|+|y_e' (t)|dt.
            \end{align*}
            Since $ x_e$ is non-decreasing on $[0,1]$ 
            \begin{align*}
                &\int_{t_{k-1}}^{t_k} | x_e' (t)|dt
                =\int_{t_{k-1}}^{t_k}  x_e' (t)dt  
                = x_e(t_k) - x_e (t_{k-1})
                = |x_e (t_k)- x_e (t_{k-1})|.
            \end{align*}
            Analogously, $\int_{t_{k-1}}^{t_k} | y_e' (t)|dt=|y_e(t_{k})-y_e(t_{k-1})|$.  
        \end{proof}

        \begin{proof}[Proof of Claim \ref{claim: norm1 pi}]
            Notice that $\gamma_e(t)=(x_e(t), y_e(t))$ and $\gamma_e^\perp(t)=( x_e^\perp(t),  y_e^\perp(t))$ lie on $[0,1]\times [0,1]$ which is in the first quadrant of $\mathbb{R}^2$. Then, by the geometric property of $\|\cdot\|_1$ for projected points along the direction $(1,-1)$ in the same quadrant, we have $\|\gamma_e(t)\|_1=\|\gamma_e^\perp(t)\|_1$ for all $t\in [0,1]$.
            
            Since $ x_e$, $ y_e$ are non-decreasing functions over $[0,1]$, we also have that $x_e'\geq 0$ and  $  y_e'\geq 0$. Thus, 
            \begin{align*}
                \|\gamma_e(t)\|_1&= x_e(t)+y_e(t)=\int_0^t\ x_e'(u)+ y_e'(u) \, du=\int_0^t\|\gamma_e'(u)\|_1du.                
            \end{align*}
                        Analogously, since $\gamma_e^\perp$ is an admissible path, we have $\|\gamma_e^\perp(t)\|_1=\int_0^t\|(\gamma_e^\perp)'(u)\|_1 du $.

            Therefore,            \begin{equation*}
                \int_0^t\|\gamma_e'(u)\|_1du=\int_0^t\|(\gamma_e^\perp)'(u)\|_1du.                         
            \end{equation*}

            Taking derivative on both sides we have $\|\gamma_e'(t)\|_1=\|(\gamma_e^\perp)'(t)\|_1$.
        \end{proof}

        \begin{proof}[Proof of Claim \ref{claim: h_tilde_bound}]
            By definition, $ h (x(t),  y(t)) = |s(x(t))-\phi(y(t))|$.
            Thus,
            \begin{align*}
                &\left|h(\gamma_e^\perp(t))- h(\gamma_e(t_k))\right|=\left| \left| s(x_e^\perp(t))-\phi(y_e^\perp(t)) \right| - \left|s(x_e(t_k))-\phi(y_e(t_k)) \right|\right| \\
                &\leq\left| s(x_e^\perp(t))- s(x_e(t_k)) -\phi(y_e^\perp(t)) +\phi(y_e(t_k)) \right|\\
                &\leq   \|s'\|_\infty |x_e^\perp(t)- x_e(t_k)| + \|\phi'\|_\infty |y_e^\perp(t) - y_e(t_k) | \\
                &\leq \max\{\|s'\|_\infty,\|\phi'\|_\infty\}  \left(|x_e^\perp(t)- x_e(t_k)| + |y_e^\perp(t) - y_e(t_k) |\right)\\
                &= \max\{\|s'\|_\infty,\|\phi'\|_\infty\} \|(\gamma_e^\perp)(t)-\gamma_e(t_k)\|_1\\
                &\leq \max\{\|s'\|_\infty,\|\phi'\|_\infty\} \cdot 3\Delta
            \end{align*}
            where we have used that $\|(\gamma_e^\perp)(t)-\gamma_e(t_k)\|_1\leq 3\Delta$ because  for all $t\in(t_{k-1},t_k)$, the point $\gamma_e^\perp(t)$ is at most 2 cells apart from the vertex $v_k = \gamma_e(t_k)$. And two points in contiguous cells will always be at $L^1$-distance at most $3\Delta$. 
        \end{proof}

\subsection{The $d_T$ Divergence}\label{app: d_T}

    \begin{proof}[Proof of Lemma \ref{lem: extended graph}  under hypothesis \textbf{H} (Definition \ref{def: hyp H})]
        If $s\circ g_1=\phi\circ g_2$, then $\|s'\|_1=\|\phi'\|_1 $.
        Indeed, $s\circ g_1=\phi\circ g_2$ implies $(s'\circ g_1)g_1'=(\phi'\circ g_2)g_2'$. Using that $g_1, g_2$ are non-decreasing we get $(|s'|\circ g_1) g_1' = (|\phi'|\circ g_2) g_2'$. Thus, 
        \begin{align*}
            \|s'\|_1&=\int_0^1|s'(u)|du=\int_0^1 |s'|(g_1(t))g_1'(t)dt\\
            &=\int_0^1 |\phi'|(g_2(t))g_2'(t)dt=  \int_0^1|\phi'(u)|du=\|\phi'\|_1.
        \end{align*}
        Therefore, $(\frac{|s'|}{\|s'\|_1}\circ g_1) g_1' = (\frac{|\phi'|}{\|\phi'\|_1}\circ g_2) g_2'$. As in Section \ref{subsec: sdb}, integrating on both sides yields
         
        \begin{align}\label{eq: fundamental identity}
                F_{\frac{|s'|}{\|s'\|_1}}(g_1(u)) = F_{\frac{|\phi'|}{\|\phi'\|_1}}(g_2(u)).
        \end{align}
    
        By doing a change of variables $t=g_1(u)$, we can write
        \begin{small}
            \begin{align}
                &d_T(s,\phi) = \int_0^1 |s(t)-\phi(g^s_\phi(t))|^2 \sqrt{(g^s_\phi)'(t)}\ dt \notag\\
                &=\int_0^1 \left|s(g_1(u))-\phi\left(F_{\frac{|\phi'|}{\|\phi'\|_1}}^\dagger((F_{\frac{|s'|}{\|s'\|_1}}(g_1(u))))\right)\right|^2 \sqrt{\left(F_{\frac{|\phi'|}{\|\phi'\|_1}}^\dagger F_{\frac{|s'|}{\|s'\|_1}}\right)'(g_1(u))} \, g_1'(u)du \notag\\
                &=\int_0^1 \underbrace{\left|s(g_1(u))-\phi\left(F_{\frac{|\phi'|}{\|\phi'\|_1}}^\dagger(F_{\frac{|s'|}{\|s'\|_1}}(g_1(u)))\right)\right|^2 \sqrt{\left(F_{\frac{|\phi'|}{\|\phi'\|_1}}^\dagger \circ F_{\frac{|s'|}{\|s'\|_1}}\circ g_1\right)'(u) \, g_1'(u)}}_{\text{A(u)}}  du\label{eq: d_t mult int}
            \end{align}  
        \end{small}
        where in the last equality holds by the chain rule. 
        
        Notice that $F_{\frac{|\phi'|}{\|\phi'\|_1}}$ and $g_2$ are non-decreasing and have finitely many flat regions. Thus, the same is true for its composition $F_{\frac{|\phi'|}{\|\phi'\|_1}} \circ g_2$. Let us denote $I_i$, $i=1,\dots, N$ each flat region of the function $F_{\frac{|\phi'|}{\|\phi'\|_1}} \circ g_2$. 
        
        On the one hand, if $u \in [\bigcup I_i]^c$, then $F_{\frac{|\phi'|}{\|\phi'\|_1}}$ must be strictly right-increasing at $g_2(u)$ (in the sense of  Proposition \ref{prop: properties_gen_inv}). By Proposition \ref{prop: properties_gen_inv} and \eqref{eq: fundamental identity}, we have
        \begin{equation*}
            F_{\frac{|\phi'|}{\|\phi'\|_1}}^{\dagger}(F_{\frac{|s'|}{\|s'\|_1}}(g_1(u))) = g_2(u).
        \end{equation*}
        Therefore,
         \begin{equation}\label{eq: aux dt 1}
             s(g_1(u))=\phi(g_2(u))=\phi\left(F_{\frac{|\phi'|}{\|\phi'\|_1}}^{\dagger}(F_{\frac{|s'|}{\|s'\|_1}}(g_1(u))) \right) \quad \forall u\in \left[\bigcup I_i\right]^c.
        \end{equation}
        
        On the other hand, without loss of generality, assume $I_i$ to be an interval with non-empty interior. Then, 
        $F_{\frac{|\phi'|}{\|\phi'\|_1}}\circ g_2$ is constant
        on $I_i$. Hence, using \eqref{eq: fundamental identity}, we have that there is a constant $c_i$ such that 
        \begin{equation*}
            F_{\frac{|s'|}{\|s'\|_1}}(g_1(u)) = F_{\frac{|\phi'|}{\|\phi'\|_1}}(g_2(u)) = c_i \qquad \forall u \in I_i.
        \end{equation*}
        In particular, $F_{\frac{|s'|}{\|s'\|_1}}\circ g_1$ is constant         on $I_i$. 
        By taking the composition with $F_{\frac{|\phi'|}{\|\phi'\|_1}}^{\dagger}$, this implies that        
        $F_{\frac{|\phi'|}{\|\phi'\|_1}}^{\dagger}\circ F_{\frac{|s'|}{\|s'\|_1}}\circ g_1$ is also constant on $I_i$. Thus, 
        \begin{equation}\label{eq: 0 deriv}
         (F_{\frac{|\phi'|}{\|\phi'\|_1}}^{\dagger}\circ F_{\frac{|s'|}{\|s'\|_1}}\circ g_1)' \equiv 0  \qquad \text{on }   I_i.
        \end{equation}
        
        Finally, using \eqref{eq: aux dt 1} and \eqref{eq: 0 deriv} in \eqref{eq: d_t mult int}, we obtain          
        \begin{align*}
        d_T(s,\phi)=    \int_0^1 A(u) \ du \leq \int_{[\bigcup I_i]^c} A(u) \ du + \sum _i \int_{I_i} A(u) \ du=0  
        \end{align*}
        since the integrand $A(u)$ is zero in each region of integration. Thus, $d_T(s,\phi) = 0$ as we wanted to prove.
    \end{proof}

    In order to prove Remark \ref{rmk: graph of transport} we need the following lemma. 
    \begin{lemma}\label{lm: reparametrization}
        Let $F,G,\tilde{h}_1,\tilde{h}_2 \in \mathcal{G}$ such that $G\circ \tilde{h}_1 = F \circ \tilde{h}_2$. Then, there exists $h_1,h_2 \in \mathcal{G}$ satisfying:
        \begin{enumerate}
            \item $G\circ h_1 = F \circ h_2$.
            \item For every $(t, F^\dagger \circ G(t))\in Graph(F^\dagger \circ G)$ there exists an $u$ such that $t = h_1(u)$, $F^\dagger \circ G(t) = h_2(u)$. In other words, the graph of the function $F^\dagger \circ G$ is included in the graph of the curve $t\mapsto (h_1(t),h_2(t))$. 
        \end{enumerate} 
        
    \end{lemma}
    
    \begin{proof}
        The proof follows from noticing that $F \sim G$ according to the relation given in Proposition \ref{prop: equivalence relation} and then using the construction of $h_1, h_2$ as in the proof of the transitivity property in such proposition (with $g_2$ replaced by $G$ and $f_1$ by $F$ in \eqref{eq: aux comp}). See Remark \ref{rmk: parametrization}. 
    \end{proof}

    \begin{proof}[Proof of Remark \ref{rmk: graph of transport} under hypothesis \textbf{H} (Definition \ref{def: hyp H})]
        If $s\circ g_1=\phi\circ g_2$, then $\|s'\|_1=\|\phi'\|_1 $. By the same arguments as in Section \ref{subsec: sdb}, this implies that 
        \begin{align*}
            F_{\frac{|s'|}{\|s'\|_1}}(g_1(t)) = F_{\frac{|\phi'|}{\|\phi'\|_1}}(g_2(t)).
        \end{align*}
                Due to Lemma \ref{lm: reparametrization} above, there exist $h_1,h_2$ such that given $(t, F_{\frac{|\phi'|}{\|\phi'\|_1}}^\dagger(F_{\frac{|s'|}{\|s'\|_1}}(t)))$, there exists $u$ satisfying 
        \begin{equation}\label{eq: aux ff}
            (t, F_{\frac{|\phi'|}{\|\phi'\|_1}}^\dagger(F_{\frac{|s'|}{\|s'\|_1}}(t)))=(h_1(u),h_2(u)) . 
        \end{equation}
        This shows that the graph of $g_\phi^s = F_{\frac{|\phi'|}{\|\phi'\|_1}}^\dagger\circ F_{\frac{|s'|}{\|s'\|_1}}$ is included in the graph of the curve $(h_1,h_2)$. 

    \end{proof}

\subsection{A method for solving the classification problem \ref{def: classification problem}}\label{app: classif}

    \begin{proof}[Proof of Theorem \ref{thm: cdtw implies sdb} under hypothesis \textbf{H} (Definition \ref{def: hyp H})]\,
    
    \noindent    Assume $CDTW(s,\phi)=0$. By Theorem \ref{thm: cdtw class}, there exists $g_1,g_2\in \mathcal{G}$ such that
                    $s\circ g_1=\phi\circ g_2.$   
        Thus, the conclusion $d_T(s,\phi)=0$ follows from Lemma \ref{lem: extended graph}.
    \end{proof}

    \begin{proof}[Proof of Theorem \ref{thm: CDTW guarantees} ]
       If $s\in\mathbb{S}^{(c)} = \bigcup_{m=1}^{M_c} \mathbb{S}_{\phi_m^{(c)}}$ for some class $c = c_0$, there exists $m_0\in \{1,\dots, M_{c_0}\}$ such that $s\in \mathbb{S}_{\phi_{m_0}^{(c_0)}}$ and       
       $s\sim \phi_{m_0}^{(c_0)}$. 
       By Proposition \ref{prop: equivalence relation}, it also hols that $s\sim s_{m_0}^{(c_0)}$, where $s_{m_0}^{(c_0)}$ is an arbitrary sample of the atomic class $\mathbb{S}_{\phi_{m_0}^{(c_0)}}$. That is, there exists $g_1,g_2\in \mathcal{G}$ such that $s\circ g_1=s_{m_0}^{(c_0)}\circ g_2$. Then, by Theorem \ref{thm: cdtw class}, $CDTW(s,s_{m_0}^{(c_0)})=0$.

       Since each atomic class is an equivalence class, then $s$ does not belong to any other atomic class. Thus, $CDTW(s,s_m^{(c)}) > 0$ for all $m,c\neq m_0,c_0$. 
       
       Then,
        \begin{equation} 
            \min_{m,c} CDTW(s,s_m^{(c)})=0,
        \end{equation}
       and the minimum is attained at $m_0,c_0$ with $c_0$ the class of $s$.
       Therefore, CDTW perfectly recovers the class of signal $s$.   
       
       Finally, the claim about $d_T$ follows from noticing that $d_T(s,\phi)\geq 0$ for every $\phi$ and $d_T(s,s_{m_0}^{(c_0)}) = 0$ since $CDTW(s,s_{m_0}^{(c_0)}) = 0$ due to Theorem \ref{thm: cdtw implies sdb}.

    \end{proof}

\section{Appendix: Generalized Inverse Function}\label{app: gen inv}

    Let $F:[a,b] \to [c,d]$  be a non-decreasing, right continuous function. Its generalized inverse $F^\dagger :[c,d] \to [a,b]$, see Definition \ref{def: gen inv}, can be visualized just as a regular inverse in the sense that its graph reflects the graph of $F$ through the line $y=x$. Nevertheless, constant regions are reflected as discontinuity jumps and jumps as constant regions. The value of the function at each jump is always the value at the high value since it must be non-decreasing and right continuous. See Figure \ref{fig: gen inv}.

\begin{figure}[ht!]
    \centering
    \includegraphics[width=0.35\linewidth]{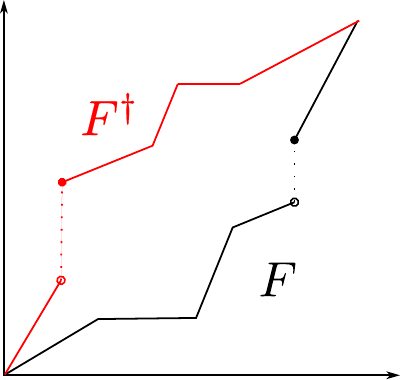}
    \caption{Graph of a right-continuous non-decreasing function $F$ (black) and its generalized inverse $F^\dagger$(red).}
    \label{fig: gen inv}
\end{figure}

\begin{proposition} \label{prop: properties_gen_inv}     
    (Properties of generalized inverses)
    Let $F:[a,b]\to[c,d]$ be a non-decreasing, right continuous function and satisfying $F(b) = d$. Then,
    \begin{enumerate}
        \item $(F^\dagger)^\dagger = F$.
        \item $F^\dagger\circ F (x) \geq x$.
        \item If $F(x)$ is strictly right-increasing at $x$, in the sense that $$\forall \varepsilon >0, \,  F(x) < F(x+\varepsilon), $$ then $F^\dagger\circ F (x) = x$. 
    \end{enumerate}    
\end{proposition}

\begin{proof} For the proof of part 1 we refer to the Appendix of the article \cite{aldroubi2021signed}. For the proof of parts 2 and 3, by definition we have 
        $$F^\dagger\circ F(x) = F^\dagger(F(x))  = \inf\{z\in [a,b]: F(z)>F(x)\}.$$ 
    Since $F$ is non-decreasing we have that $$\{ z \in [a,b]: F(z)>F(x) \} \subseteq \{z \in [a,b]: z>x \}.$$
    With the equality holding only if $F$ is strictly increasing at $x$.
    Thus, 
        $$F^\dagger (F(x)) = \inf\{z\in [a,b]: F(z)>F(x)\} \geq \inf \{z \in [a,b]: z>x\} = x.$$ 
    With the inequality being an equality only if $F$ is strictly increasing at $x$.
\end{proof}

\bibliographystyle{IEEEtran}
\bibliography{references}

\end{document}

%% file: small_table_details.tex
\begin{table}[ht!]
    \centering 
    \resizebox{\textwidth}{!}
    {
        \begin{tabular}{lrrrlll}
        \toprule
        Dataset & Train size & Test size & Classes & Euclidean & $d_T$ & DTW \\
        \midrule
        BME & 30 & 150 & 3 & 82.67 & \textbf{99.33} & 88.67 \\
        Chinatown & 20 & 343 & 2 & 92.87 & \textbf{95.86} & 93.34 \\
        BirdChicken & 20 & 20 & 2 & 55.00 & 70.00 & \textbf{75.00} \\
        DiatomSizeReduction & 16 & 306 & 4 & 86.32 & \textbf{93.67} & 90.64 \\
        FreezerRegularTrain & 150 & 2850 & 2 & 80.49 & \textbf{87.89} & 87.30 \\
        FreezerSmallTrain & 28 & 2850 & 2 & 67.58 & \textbf{75.86} & 69.65 \\
        LargeKitchenAppliances & 375 & 375 & 3 & 49.33 & 78.67 & \textbf{80.53} \\
        Lightning2 & 60 & 61 & 2 & 74.84 & \textbf{82.95} & 79.11 \\
        Lightning7 & 70 & 73 & 7 & 52.32 & 63.49 & \textbf{74.59} \\
        StarLightCurves & 1000 & 8236 & 3 & 89.59 & \textbf{93.07} & 91.48 \\
        Symbols & 25 & 995 & 6 & 89.71 & \textbf{95.06} & 94.89 \\
        Trace & 100 & 100 & 4 & 78.25 & \textbf{100.00} & 99.11 \\
        TwoLeadECG & 23 & 1139 & 2 & 74.73 & \textbf{93.33} & 92.72 \\
        TwoPatterns & 1000 & 4000 & 4 & 90.63 & 97.34 & \textbf{100.00} \\
        UMD & 36 & 144 & 3 & 80.56 & \textbf{95.14} & 84.72 \\
        \bottomrule
        \end{tabular}
    }
    \caption{Euclidean, $d_T$, and DTW metrics evaluated on datasets from the UCR time series classification archive \cite{dau2019ucr}.}
    \label{tab: small_table_details}
\end{table}

%% file: main.bbl
\begin{thebibliography}{10}
\providecommand{\url}[1]{#1}
\csname url@samestyle\endcsname
\providecommand{\newblock}{\relax}
\providecommand{\bibinfo}[2]{#2}
\providecommand{\BIBentrySTDinterwordspacing}{\spaceskip=0pt\relax}
\providecommand{\BIBentryALTinterwordstretchfactor}{4}
\providecommand{\BIBentryALTinterwordspacing}{\spaceskip=\fontdimen2\font plus
\BIBentryALTinterwordstretchfactor\fontdimen3\font minus \fontdimen4\font\relax}
\providecommand{\BIBforeignlanguage}[2]{{%
\expandafter\ifx\csname l@#1\endcsname\relax
\typeout{** WARNING: IEEEtran.bst: No hyphenation pattern has been}%
\typeout{** loaded for the language `#1'. Using the pattern for}%
\typeout{** the default language instead.}%
\else
\language=\csname l@#1\endcsname
\fi
#2}}
\providecommand{\BIBdecl}{\relax}
\BIBdecl

\bibitem{rebbapragada2009finding}
U.~Rebbapragada, P.~Protopapas, C.~E. Brodley, and C.~Alcock, ``Finding anomalous periodic time series: An application to catalogs of periodic variable stars,'' \emph{Machine learning}, vol.~74, pp. 281--313, 2009.

\bibitem{lara2012survey}
O.~D. Lara and M.~A. Labrador, ``A survey on human activity recognition using wearable sensors,'' \emph{IEEE communications surveys \& tutorials}, vol.~15, no.~3, pp. 1192--1209, 2012.

\bibitem{fehske2005new}
A.~Fehske, J.~Gaeddert, and J.~H. Reed, ``A new approach to signal classification using spectral correlation and neural networks,'' in \emph{First IEEE International Symposium on New Frontiers in Dynamic Spectrum Access Networks, 2005. DySPAN 2005.}\hskip 1em plus 0.5em minus 0.4em\relax IEEE, 2005, pp. 144--150.

\bibitem{abdeljaber20181}
O.~Abdeljaber, O.~Avci, M.~S. Kiranyaz, B.~Boashash, H.~Sodano, and D.~J. Inman, ``1-d cnns for structural damage detection: Verification on a structural health monitoring benchmark data,'' \emph{Neurocomputing}, vol. 275, pp. 1308--1317, 2018.

\bibitem{BISVP08}
G.~Barrenetxea, F.~Ingelrest, G.~Schaefer, M.~Vetterli, O.~Couach, and M.~Parlange, ``Sensorscope: Out-of-the-box environmental monitoring,'' in \emph{2008 International Conference on Information Processing in Sensor Networks (IPSN 2008)}, 2008, pp. 332--343.

\bibitem{berkaya2018survey}
S.~K. Berkaya, A.~K. Uysal, E.~S. Gunal, S.~Ergin, S.~Gunal, and M.~B. Gulmezoglu, ``A survey on ecg analysis,'' \emph{Biomedical Signal Processing and Control}, vol.~43, pp. 216--235, 2018.

\bibitem{rubaiyat2024end}
A.~H.~M. Rubaiyat, S.~Li, X.~Yin, M.~Shifat-E-Rabbi, Y.~Zhuang, and G.~K. Rohde, ``End-to-end signal classification in signed cumulative distribution transform space,'' \emph{IEEE Transactions on Pattern Analysis and Machine Intelligence}, 2024.

\bibitem{subasi2010eeg}
A.~Subasi and M.~I. Gursoy, ``Eeg signal classification using pca, ica, lda and support vector machines,'' \emph{Expert systems with applications}, vol.~37, no.~12, pp. 8659--8666, 2010.

\bibitem{bagnall2017great}
A.~Bagnall, J.~Lines, A.~Bostrom, J.~Large, and E.~Keogh, ``The great time series classification bake off: a review and experimental evaluation of recent algorithmic advances,'' \emph{Data mining and knowledge discovery}, vol.~31, pp. 606--660, 2017.

\bibitem{middlehurst2024bake}
M.~Middlehurst, P.~Sch{\"a}fer, and A.~Bagnall, ``Bake off redux: a review and experimental evaluation of recent time series classification algorithms,'' \emph{Data Mining and Knowledge Discovery}, pp. 1--74, 2024.

\bibitem{fulcher2017hctsa}
B.~D. Fulcher and N.~S. Jones, ``hctsa: A computational framework for automated time-series phenotyping using massive feature extraction,'' \emph{Cell systems}, vol.~5, no.~5, pp. 527--531, 2017.

\bibitem{christ2018time}
M.~Christ, N.~Braun, J.~Neuffer, and A.~W. Kempa-Liehr, ``Time series feature extraction on basis of scalable hypothesis tests (tsfresh--a python package),'' \emph{Neurocomputing}, vol. 307, pp. 72--77, 2018.

\bibitem{lubba2019catch22canonicaltimeseriescharacteristics}
\BIBentryALTinterwordspacing
C.~H. Lubba, S.~S. Sethi, P.~Knaute, S.~R. Schultz, B.~D. Fulcher, and N.~S. Jones, ``catch22: Canonical time-series characteristics,'' 2019. [Online]. Available: \url{https://arxiv.org/abs/1901.10200}
\BIBentrySTDinterwordspacing

\bibitem{middlehurst2022freshprince}
M.~Middlehurst and A.~Bagnall, ``The freshprince: A simple transformation based pipeline time series classifier,'' in \emph{International Conference on Pattern Recognition and Artificial Intelligence}.\hskip 1em plus 0.5em minus 0.4em\relax Springer, 2022, pp. 150--161.

\bibitem{wang2016timeseriesclassificationscratch}
\BIBentryALTinterwordspacing
Z.~Wang, W.~Yan, and T.~Oates, ``Time series classification from scratch with deep neural networks: A strong baseline,'' 2016. [Online]. Available: \url{https://arxiv.org/abs/1611.06455}
\BIBentrySTDinterwordspacing

\bibitem{cheng2021time}
W.~X. Cheng, P.~N. Suganthan, and R.~Katuwal, ``Time series classification using diversified ensemble deep random vector functional link and resnet features,'' \emph{Applied Soft Computing}, vol. 112, p. 107826, 2021.

\bibitem{wu2019resnet}
J.~Wu, Z.~Zhang, Y.~Ji, S.~Li, and L.~Lin, ``A resnet with ga-based structure optimization for robust time series classification,'' in \emph{2019 IEEE International Conference on Smart Manufacturing, Industrial \& Logistics Engineering (SMILE)}.\hskip 1em plus 0.5em minus 0.4em\relax IEEE, 2019, pp. 69--74.

\bibitem{ismail2020inceptiontime}
H.~Ismail~Fawaz, B.~Lucas, G.~Forestier, C.~Pelletier, D.~F. Schmidt, J.~Weber, G.~I. Webb, L.~Idoumghar, P.-A. Muller, and F.~Petitjean, ``Inceptiontime: Finding alexnet for time series classification,'' \emph{Data Mining and Knowledge Discovery}, vol.~34, no.~6, pp. 1936--1962, 2020.

\bibitem{10020496}
A.~Ismail-Fawaz, M.~Devanne, J.~Weber, and G.~Forestier, ``Deep learning for time series classification using new hand-crafted convolution filters,'' in \emph{2022 IEEE International Conference on Big Data (Big Data)}, 2022, pp. 972--981.

\bibitem{sakoe1978dynamic}
H.~Sakoe and S.~Chiba, ``Dynamic programming algorithm optimization for spoken word recognition,'' \emph{IEEE transactions on acoustics, speech, and signal processing}, vol.~26, no.~1, pp. 43--49, 1978.

\bibitem{benamou2000computational}
J.-D. Benamou and Y.~Brenier, ``A computational fluid mechanics solution to the monge-kantorovich mass transfer problem,'' \emph{Numerische Mathematik}, vol.~84, no.~3, pp. 375--393, 2000.

\bibitem{buchin2007computability}
M.~E. Buchin, ``On the computability of the fr{\'e}chet distance between triangulated surfaces,'' Ph.D. dissertation, PhD thesis, Dept. of Comput. Sci., Freie Universität, Berlin, 2007.

\bibitem{buchin2009exact}
K.~Buchin, M.~Buchin, and Y.~Wang, ``Exact algorithms for partial curve matching via the fr{\'e}chet distance,'' in \emph{Proceedings of the twentieth annual ACM-SIAM symposium on Discrete algorithms}.\hskip 1em plus 0.5em minus 0.4em\relax SIAM, 2009, pp. 645--654.

\bibitem{mueen2018speeding}
A.~Mueen, N.~Chavoshi, N.~Abu-El-Rub, H.~Hamooni, A.~Minnich, and J.~MacCarthy, ``Speeding up dynamic time warping distance for sparse time series data,'' \emph{Knowledge and Information Systems}, vol.~54, pp. 237--263, 2018.

\bibitem{li2018similarity}
H.~Li and C.~Wang, ``Similarity measure based on incremental warping window for time series data mining,'' \emph{IEEE Access}, vol.~7, pp. 3909--3917, 2018.

\bibitem{Villani2003Topics}
\BIBentryALTinterwordspacing
C.~Villani, \emph{Topics in Optimal Transportation}.\hskip 1em plus 0.5em minus 0.4em\relax American Mathematical Society, 2003. [Online]. Available: \url{http://www.ams.org/gsm/058}
\BIBentrySTDinterwordspacing

\bibitem{Villani2009Optimal}
\BIBentryALTinterwordspacing
------, \emph{Optimal transport: old and new}.\hskip 1em plus 0.5em minus 0.4em\relax Springer, 2009. [Online]. Available: \url{https://www.cedricvillani.org/sites/dev/files/old_images/2012/08/preprint-1.pdf}
\BIBentrySTDinterwordspacing

\bibitem{Santambrogio-OTAM}
F.~Santambrogio, \emph{Optimal Transport for Applied Mathematicians}, ser. Progress in Nonlinear Differential Equations and Their Applications.\hskip 1em plus 0.5em minus 0.4em\relax Birkh\"auser Boston, 2015, vol.~87.

\bibitem{Thorpe2018Introduction}
\BIBentryALTinterwordspacing
M.~Thorpe, ``Introduction to optimal transport,'' \emph{Centre for Mathematical Sciences University of Cambridge}, p.~56, 2018. [Online]. Available: \url{http://www.damtp.cam.ac.uk/research/cia/files/teaching/Optimal_Transport_Notes.pdf}
\BIBentrySTDinterwordspacing

\bibitem{gold2016dynamic}
O.~Gold and M.~Sharir, ``Dynamic time warping and geometric edit distance: Breaking the quadratic barrier,'' \emph{ACM Transactions On Algorithms (TALG)}, vol.~14, no.~4, pp. 1--17, 2016.

\bibitem{dau2019ucr}
H.~A. Dau, A.~Bagnall, K.~Kamgar, C.-C.~M. Yeh, Y.~Zhu, S.~Gharghabi, C.~A. Ratanamahatana, and E.~Keogh, ``The ucr time series archive,'' \emph{IEEE/CAA Journal of Automatica Sinica}, vol.~6, no.~6, pp. 1293--1305, 2019.

\bibitem{tutorial}
Imaging and data~science lab. Pytranskit - dt divergence tutorial. \url{https://github.com/rohdelab/PyTransKit/blob/master/tutorials/16_dT_divergence_tutorial.ipynb}.

\bibitem{klaren2020continuous}
K.~Klaren, K.~Buchin, and H.~van~de Wetering, ``Continuous dynamic time warping for clustering curves,'' Ph.D. dissertation, Master’s thesis, Eindhoven University of Technology, 2020.

\bibitem{aldroubi2021signed}
A.~Aldroubi, R.~D. Martin, I.~Medri, G.~K. Rohde, and S.~Thareja, ``The signed cumulative distribution transform for 1-d signal analysis and classification,'' \emph{Foundations of Data Science}, vol.~4, no.~1, pp. 137--163, 2022.

\end{thebibliography}
